\pdfoutput=1
\documentclass{article}

     \usepackage[preprint,nonatbib]{neurips_2019}

\usepackage[utf8]{inputenc} % allow utf-8 input
\usepackage[T1]{fontenc}    % use 8-bit T1 fonts
\usepackage{hyperref}       % hyperlinks
\usepackage{url}            % simple URL typesetting
\usepackage{booktabs}       % professional-quality tables
\usepackage{amsfonts}       % blackboard math symbols
\usepackage{nicefrac}       % compact symbols for 1/2, etc.
\usepackage{microtype}      % microtypography
\usepackage{color}
\usepackage{amsmath,amssymb}
\usepackage{amscd,amsthm}
\usepackage{amscd}

\usepackage{graphicx} % more modern

\usepackage{wrapfig}

\usepackage{multirow}
\usepackage{accents}
\usepackage[usenames,dvipsnames]{xcolor}
\usepackage{enumitem}

\usepackage{algorithm}
\usepackage{algorithmic}

\theoremstyle{plain}
\newtheorem{theorem}{Theorem}%[section]
\newtheorem{lemma}{Lemma}

\newtheorem{proposition}{Proposition}

\theoremstyle{definition}
\newtheorem{definition}{Definition}
\newcommand{\A}{\mathcal{A}}
\newcommand{\T}{\mathfrak{T}}
\newcommand{\RR}{\mathfrak{R}}
\newcommand{\LL}{\mathfrak{L}}

\newcommand{\Reg}{\mathrm{Reg}}

\newcommand{\SReg}{\mathrm{SReg}}
\newcommand{\Sur}{\mathrm{Sur}}

\newcommand{\n}{\mathfrak{n}}
\newcommand{\m}{\mathfrak{m}}
\newcommand{\rt}{\mathfrak{r}}
\newcommand{\lf}{\mathfrak{l}}

\newcommand{\p}{\mathrm{p}}

\newcommand{\Ind}{\mathbb{I}}
\newcommand{\fa}{\:\forall}

\newcommand{\g}{\gamma}
\newcommand{\ve}{\varepsilon}
\newcommand{\e}{\epsilon}

\newcommand{\gb}{{\boldsymbol\g}}

\newcommand{\argmax}{\mathop{\mathrm{argmax}}}

\newcommand{\ASet}{\mathbf{A}}
\newcommand{\pre}{\mathtt{pre}} 
\newcommand{\treeroot}{\mathfrak{e}}

\newcommand{\Bset}{\mathbb{M}}
\newcommand{\bBset}{\overline{\Bset}}
\newcommand{\Sset}{\mathbb{S}}

\newcommand{\vb}{\mathbf{v}}
\newcommand{\bb}{\mathbf{b}}
\newcommand{\pb}{\mathbf{p}}
\newcommand{\bm}{\overline{m}}
\newcommand{\bv}{\overline{v}}
\newcommand{\bbv}{\overline{\bv}}
\newcommand{\bp}{\overline{p}}
\newcommand{\bpb}{\overline{\mathbf{p}}}
\newcommand{\ba}{\overline{a}}
\newcommand{\bab}{\overline{\mathbf{a}}}
\newcommand{\ab}{\mathbf{a}}

\newcommand{\lb}{\mathbf{l}}
\newcommand{\qb}{\mathbf{q}}

\newcommand{\sig}{\sigma}

\newcommand{\betab}{{\boldsymbol\beta}}

\newcommand{\HH}{\mathbb{H}}

\newcommand{\StrartSet}{\mathfrak{S}}

\newcommand{\TT}{\mathcal{T}}
\newcommand{\II}{\mathcal{I}}

\newcommand{\diva}{\mathrm{div}}

\newcommand{\osig}{\mathring{\sigma}}

\newcommand{\ob}{\mathring{b}}
\newcommand{\obb}{\mathring{\mathbf{b}}}

\newcommand{\ind}{\mathrm{ind}}
\newcommand{\dev}{\mathrm{dev}}
\newcommand{\RPPA}{\mathrm{RPPA}}

\newcommand{\sr}{\mathrm{sr}}

\newcommand{\eset}{\varnothing}
\newcommand{\Expect}{\mathbb{E}}

\title{Reserve Pricing in Repeated Second-Price Auctions with Strategic Bidders}

\author{%
  Alexey Drutsa\\
  Yandex; MSU\\
  Moscow, Russia\\
  \texttt{adrutsa@yandex.ru} \\
}

\begin{document}

\maketitle

\begin{abstract}
We study revenue optimization learning algorithms for repeated second-price auctions with reserve where a seller interacts with multiple strategic bidders each of which holds a fixed private valuation for a good and seeks to maximize his expected future cumulative discounted surplus.	
We propose a novel algorithm that  has strategic regret upper bound of $O(\log\log T)$ for worst-case valuations.
%This pricing is based on our novel technique to transform an algorithm designed for the setup with a single buyer to the one for the studied multi-buyer case. 
This pricing is based on our novel transformation that upgrades an algorithm designed for the setup with a single buyer to the multi-buyer case. 
We provide theoretical guarantees on the ability of a transformed algorithm to  learn the valuation of a strategic buyer, which has uncertainty about 
the future due to the presence of rivals.
\end{abstract}

\section{Introduction}
\label{sec_Intro}

Revenue maximization is one of fundamental development directions  in  major Internet companies that have their  own online advertising platforms~\cite{2014-WWW-Gomes,2015-ManagSci-Balseiro,2014-KDD-Agarwal,2017-WWW-Drutsa,2018-IJGT-Hummel}.
Most part of ad inventory is sold via widely applicable second price auctions~\cite{2013-IJCAI-He,2014-ICML-Mohri} and their generalizations like GSP~\cite{2007-IJIO-Varian,2009-AER-Varian,2014-AER-Varian,2014-ECRA-Sun}.
Adjustment of reserve prices plays a central role in revenue optimization here: their proper setting is studied both by game-theoretical methods~\cite{1981-MOR-Myerson,agrawal2018robust} and by machine learning approaches~\cite{2007-Book-Nisan,2013-SODA-Cesa-Bianchi, 2014-ICML-Mohri,2016-WWW-Paes}.

In our work, we focus on a scenario where the seller \emph{repeatedly} interacts through a second-price auction with $M$ strategic bidders (referred to as buyers as well).
Each buyer participates in each round of this game, holds a \emph{fixed} private valuation for a good (e.g., an ad
space), and seeks to maximize his expected future discounted surplus given his beliefs about the behaviors of other bidders.
The seller applies a deterministic online learning algorithm, which is announced to the buyers in advance and, in each round, selects  individual reserve prices based on the previous bids of the buyers.
The seller's goal is to maximize her  revenue over a finite horizon~$T$  through \emph{regret} minimization for \emph{worst-case} valuations of the bidders~\cite{2014-NIPS-Mohri,2018-ICML-Drutsa}.
Thus, the seller  seeks for a \emph{no-regret} pricing algorithm.
%, i.e., with a sublinear regret on~$T$.

To the best of our knowledge, no existing study investigated  \emph{worst-case regret optimizing} algorithms that set reserve prices in repeated second-price auctions with strategic bidders whose \emph{valuation is private, but fixed} over all rounds.
However, our setting constitutes a natural generalization of the well-studied $1$-buyer setup of repeated posted-price auctions\footnote{In particular, when $M=1$, our auction in a round reduces to a posted-price one: the bidder has no rivals and his decision is thus  binary (to accept or to reject a currently offered price).} (RPPA)~\cite{2013-NIPS-Amin,2014-NIPS-Mohri} to the scenario of multiple buyers in a second-price auction. In the RPPA setting, there are optimal algorithms~\cite{2017-WWW-Drutsa,2017-ArXiV-Drutsa,2018-ICML-Drutsa} that have tight strategic regret bound of $\Theta(\log\log T)$.
This bound follows from an ability of the seller to upper bound the buyer valuation even if he lies when rejecting a price~\cite[Prop.2]{2017-WWW-Drutsa}. 
This ability strongly exploits that the buyer knows in advance the outcomes of a current and all future rounds since he  has complete information due to the absence of rivals.
In our multi-bidder scenario, this does not hold: a bidder has incomplete information and is thus uncertain about the future. Hence, the theoretical guarantees could not be directly ported to our scenario when trying straightforwardly apply the optimal $1$-buyer RPPA algorithms.

In our study, we propose  a novel algorithm that can be applied against our strategic buyers  with regret upper bound of $O(\log\log T)$  (Th.~\ref{th_divPRRFES_sreg_uppbound}) and constitutes \emph{the main contribution of our work}.
We also introduce a novel transformation of a RPPA algorithm that  maps it to a multi-buyer pricing and is based on a simple but crucial idea of cyclic elimination of all bidders except one in each round (Sec.\ref{sec_divClass}). 
Construction and analysis of the  proposed algorithm and transformation have required introduction of novel techniques, which are contributed by our work as well. 
They include (a)~the method to locate the valuation of a strategic buyer in a played round under his uncertainty about the future (Prop.~\ref{prop_reject_bound_divalg}); 
(b)~the decomposition of strategic regret into the regret of learning the individual valuations and the deviation regret of learning which bidder has the maximal valuation (Lemma~\ref{lemma_divalg_reg_decompose}); and
(c)~the approach to learn the highest-valuation bidder with deviation regret of $O(1)$ w.r.t.\ $T$ (Lemma~\ref{lemma_divPRRFES_devreg_uppbound}).

\section{Preliminaries: setup, background, and overview of results}
\label{sec_Fwk}
%\subsection{Setup of Repeated Second-Price Auctions}
%\label{subsec_Setup}
{\bf Setup of Repeated Second-Price Auctions.}
We study the following mechanism of \emph{repeated second-price auctions}.
Namely, \emph{the auctioneer} repeatedly proposes goods (e.g., advertisement opportunities) to $M$ \emph{bidders} (whose set is denoted by $\Bset := \{1,\ldots,M\}, M\in\mathbb{N}$) over $T$ rounds: one good per round. From here on the following terminology is used as well: \emph{the seller} for the auctioneer, a \emph{buyer} for a bidder, and \emph{the time horizon} for the number of rounds $T$.
Each bidder $m\in\Bset$ holds a \emph{fixed private valuation} $v^m \in [0, 1]$ for a good, i.e., the valuation $v^m$ is equal for goods offered in all rounds and is unknown  
%both
 to the seller.
% and to all other bidders $\Bset^{-m}:=\Bset\setminus\{m\}$. 
The vector of valuations of all bidders is denoted by $\vb := \{v^m\}_{m=1}^{M}$.
%\footnote{The full table of notations can be find in App.~C of Supp.Mat.}.

%{\bf TODO: need definition $\Bset^{-m}:=\Bset\setminus\{m\}$ smwhr }

In each round $t\!\in\!\{1,\ldots,T\}$, for each bidder $m\!\in\!\Bset$, the seller sets a personal reserve price $p^m_t$, and the buyer $m$ (knowing $p^m_t$) submits a sealed bid of $b^m_t$. 
Given the reserve prices $\pb_t\!\!:=\!\{p^m_t\}_{m=1}^{M}$ and  the bids $\bb_t\!\!:=\! \{b^m_t\}_{m=1}^{M}$, the standard allocation and payment rules of a second price auction are applied~\cite{2016-WWW-Paes}:
(a)~for each bidder $m\!\in\!\Bset$, we check whether he bids over his reserve price or not, 
$a^m_t\!\!:=\!\!\Ind_{\{b^m_t\ge p^m_t\}}$,
%\footnote{$\Ind_{B}$ is the indicator: $\Ind_{B} = 1$, when $B$ holds, and $0$, otherwise.},
obtaining the set $\Bset_t\!\!:=\!\!\{m\!\in\!\Bset \mid a^m_t \!=\! 1\}$ of actual bidder-participants;
(b) if $\Bset_t \neq \eset$, 
the good is allocated to the winning bidder  $\bm_t\!:=\!\argmax_{m\in\Bset_t}b^m_t$ (if a tie, choose randomly)
%\footnote{If there is a tie, then we choose one of the winners randomly.} 
who pays $\bp_t\!:=\!\max\{p^{\bm_t}_t, \max_{m\in\Bset_t\setminus\{\bm_t\}} b^m_t \}$ to the seller.
(c)~if $\Bset_t \!=\! \eset$, the current good disappears and no payment is transferred.
Further we  use the following notations for allocation indicators, payments, and their vectors: $\ba_t \!:=\! \Ind_{\{\Bset_t \neq \eset\}}$, $\ba^m_t\! :=\!  \Ind_{\{\Bset_t \neq \eset \& m = \bm_t\}}$, $\bp^m_t\! :=\! \ba^m_t\bp_t$, $\ab_t\!:=\! \{a^m_t\}_{m=1}^{M}$,  $\bab_t\!:=\! \{\ba^m_t\}_{m=1}^{M}$, and $\bpb_t\!:=\! \{\bp^m_t\}_{m=1}^{M}$.
% where $\Ind_{B}$ is the indicator of a set $B$; 
The summary on all notations is in App.~C.
%\footnote{Mnemonic notations: {\bf boldface} for a vector over bidders and $\overline{\mathrm{bar}}$ for terms associated with  auction outcomes.}.

%$\Ind_{B}$ is the indicator: $\Ind_{B} = 1$, when $B$ holds, and $0$, otherwise.

%The fact of allocation $\ba_t$ and the selling price $\bp_t$ become known to all bidders after the current round $t$.
%{\bf TODO: do we really need knowledge of these two things???? See the literature on repeated SP // ROughtgarden?? Seem that we need remove it!!!}

Thus, the seller applies a \emph{(pricing) algorithm} $\A$ that sets reserve prices $\pb_{1:T}:=\{\pb_t\}_{t=1}^{T}$ in response to the buyers' bids $\bb_{1:T}:=\{\bb_t\}_{t=1}^{T}$.
%\footnote{$x_{t_1:t_2}\!\!=\!\!\{x_t\}_{t=t_1}^{t_2}$ denotes a part of a time series $\{x_t\}_{t=1}^{T}$.}. 
We consider the deterministic online learning case when the reserve price $p^m_t$ for a bidder $m\in\Bset$ in a round $t\in\{1,\ldots,T\}$ can depend only on bids $\bb_{1:t-1}$ of all bidders during the previous rounds and, possibly, the horizon~$T$.
%\footnote{The key algorithm proposed in our work is independent of  $T$.}.
% Following~\cite{2017-WWW-Drutsa}, we are studying algorithms that do not depend on the horizon $T$ since it is very natural in practice (e.g., of ad exchanges) that the seller does not know in advance the number of rounds~$T$ that the buyer wants to interact with him. 
Let $\ASet_{M}$ be the set of such algorithms.
Hence, given a pricing algorithm $\A\in\ASet_{M}$, the buyers' bids $\bb_{1:T}$ uniquely define  the corresponding price sequence $\{\pb_t\}_{t=1}^{T}$, which, in turn, determines  the seller's total revenue $\sum_{t=1}^T \ba_t\bp_t$.
This revenue is usually compared to the revenue that would have been earned by offering the highest valuation $\bv:=\max_{m\in\Bset}v^m$ if the valuations $\vb=\{v^m\}_{m=1}^{M}$ were known in advance to the seller~\cite{2013-NIPS-Amin,2017-WWW-Drutsa}. This leads to the notion of the \emph{regret} of the algorithm $\A$:
%that faced $M$ buyers with the valuations $\vb\in[0,1]^M$ submitting bids $\bb_{1:T}$  over $T$ rounds as
$
\Reg(T,\A,\vb,\bb_{1:T}):= \sum_{t=1}^T(\bv - \ba_t\bp_t).
$

Following a standard assumption in mechanism design that matches the practice in ad exchanges~\cite{2014-NIPS-Mohri,2018-ICML-Drutsa}, the seller's pricing algorithm $\A$ \emph{is announced to the buyers in advance}. A bidder can then act strategically against this algorithm. 
In contrast to the case of one bidder ($M=1$), where the buyer can get an optimal behavior in advance, and the repeated mechanism reduces thus to a two-stage game~\cite{2013-NIPS-Amin,2014-NIPS-Mohri,2017-WWW-Drutsa}; in our setting, a bidder has incomplete information since he may not know the valuations and behaviors of the other bidders.
Therefore, in order to model buyer strategic behavior under this uncertainty,  we assume that, in each round $t$, each buyer optimizes  his utility on subgame of future rounds given the available history of previous rounds and his beliefs about the other buyers.

Formally, in a round $t$, given the seller's pricing algorithm $\A$, a strategic buyer $m\in\Bset$ observes a history $h^m_{t}\!\!:=\!\!(b^m_{1:t-1},p^m_{1:t},\ba^m_{1:t-1},\bp^m_{1:t-1})$ available to him  and derives his \emph{optimal bid} $\ob^m_t$ from a  (possibly mixed) strategy  $\sig\in\StrartSet_T$\footnote{\emph{A buyer strategy} is a map $\sigma : \HH_{1:T}  \to \mathbb{R}_+$ that maps any history $h\in\HH_t$  in a round~$t$ to a bid $\sigma(h)\in\mathbb{R}_+$,  where $\HH_{1:T}:=  \sqcup_{t=1}^{T} \HH_t$ and $\HH_t:=\mathbb{R}_+^{t-1}\times\mathbb{R}_+^{t}\times\mathbb{Z}_2^{t-1}\times\mathbb{R}_+^{t-1}$.	Let $\StrartSet_T$ denote the set of all possible strategies.} 
that maximizes his future $\g_m$-discounted surplus:
\vspace{-0.5mm}
\begin{equation}
\label{eq_surplus}
\Sur_{t:T}(\A,\g_m,v^m,h^m_t,\beta^m,\sig)=
\Expect\Big[\sum_{s=t}^T\g^{s-1}_m\ba^m_s(v^m - \bp^m_s)\mid h^m_t,\sig, \beta^m\Big],
\end{equation}
where $\g_m\in(0,1]$ is \emph{the discount rate}\footnote{Note that only buyer utilities are discounted over time, what is motivated by real-world markets as online advertising where sellers are far more willing to wait for revenue than buyers are willing to wait for goods~\cite{2014-NIPS-Mohri,2018-ICML-Drutsa}.} of the bidder $m$. The expectation in Eq.~(\ref{eq_surplus}) is taken over all possible continuations of the history $h^m_{t}$ w.r.t.\ a strategy $\sig\in\StrartSet_T$ of the buyer $m$ and  his beliefs $\beta^m$ about the strategies of the other bidders $\Bset^{-m}\!:=\!\Bset\!\setminus\!\{m\}$\footnote{So,  $\sig$ and  $\beta^m$ determine the future outcomes $\ba^m_s$ and $\bp^m_s$, that are thus random variables.}. The buyer $m$ assumes that the other bidders are strategic in the sense described above as well, what is taken into account in the beliefs $\beta^m$\footnote{In our setup, we do not require that the strategies actually used by the buyers $\Bset^{-m}$ match with the buyer $m$'s beliefs $\beta^{m}$ (an equilibrium requirement), because our results hold without this requirement.}.
When $T$ rounds has been played, 
%let us denote the realized optimal bid by  $\ob^m_t := \osig^m_{h^m_t}(h^m_t,v^m)$ for each round $t$ and each bidder $m$.
let $\obb_t \!:=\!\{ \ob^m_t\}_{m=1}^M$ be the optimal bids that depend~on $(T,\A,\vb,\gb,\betab)$, where $\gb \!= \!\{\g_m\}_{m=1}^M$ and $\betab \!=\! \{\beta_m\}_{m=1}^M$.
We define \emph{the strategic regret} of the algorithm $\A$ that faced $M$ strategic buyers with valuations $\vb\!\in\![0,1]^M$ and beliefs $\betab$ over $T$ rounds as 
\begin{equation*}
\SReg(T,\A,\vb,\gb,\betab)\!:=\!\Reg\big(T,\A,\vb,\obb_{1:T}(T,\A,\vb,\gb,\betab)\big).
\end{equation*}
In  our setting, following~\cite{2013-NIPS-Amin,2014-NIPS-Mohri,2017-WWW-Drutsa,2018-ICML-Drutsa}, we seek for  algorithms  that  attain $o(T)$ strategic regret for  the \emph{worst-case} valuations $\vb\in[0,1]^M$. 
Formally, an algorithm $\A$ is said to be a \emph{no-regret} one when $\sup_{\vb\in[0,1]^M, \betab}\SReg(T,\A,\vb,\gb,\betab)=o(T)$ in our multi-buyer case. 
The optimization goal is to find algorithms with the lowest possible strategic regret upper bound $O(f(T))$, i.e., $f(T)$ has the slowest growth as $T\rightarrow\infty$ or, alternatively, the  averaged  regret  has  the best rate of convergence to zero.

{\bf Background on pricing algorithms.}
To the best of our knowledge, there is no work studied \emph{worst-case regret optimizing} algorithms that set reserve prices in repeated second-price auctions with strategic bidders whose \emph{valuation is private, but fixed} over all rounds.
However, in the case of one bidder, $M=1$, the bidder has no rivals, and, thus, the second-price auction in a round $t$ reduces to a posted-price auction, where the buyer decision reduces to a binary action: to accept or to reject a currently offered price $p^1_t$.
Let $\ASet^{\RPPA}\subset\ASet_{1}$ be the subclass of the 1-bidder algorithms s.t.\ each reserve price $p^1_t$ depends only on the past binary  decisions $a^1_{1:t-1}$ of the buyer to get or do not get a good for a posted reserve price. 
For this subclass, all our strategic setting of repeated second-price auctions reduces to the setup of repeated posted-price auctions (RPPA) earlier introduced in~\cite{2013-NIPS-Amin}. 
%. complete information (no rival bidders), a selling price $\bp_t$ is always  the reserve price $p^1_t$ and the allocation decision the buyer just need to set his bid $b^1_t = p^1_t$.

Pricing algorithms in the strategic setup of RPPA with fixed private valuation and worst-case regret optimization were well studied last years~\cite{2013-NIPS-Amin,2014-NIPS-Mohri,2017-WWW-Drutsa,2018-ICML-Drutsa}.
%In particular,
It is known that, if the discount rate $\g \!= \!1$,  any algorithm has a linear strategic regret, i.e., the regret has lower bound $\Omega(T)$~\cite{2013-NIPS-Amin}, while, for the other cases $\g \!\in\!(0,1)$, the lower bound of $\Omega(\log\log T)$ holds~\cite{2003-FOCS-Kleinberg,2014-NIPS-Mohri}.
The first algorithm with  optimal strategic regret bound of $\Theta(\log\log T)$ was found in~\cite{2017-WWW-Drutsa}. It is~Penalized Reject-Revising Fast Exploiting Search (PRRFES), which is horizon-independent  and is based on Fast Search~\cite{2003-FOCS-Kleinberg} modified to act against a strategic buyer. The modifications include penalizations (see Def.~\ref{def_PenalNodeSeq}).
A strategic buyer either accepts the price at the first node or rejects this price in subsequent penalization ones~\cite{2014-NIPS-Mohri,2017-WWW-Drutsa}. 
PRRFES is also a right-consistent algorithm:
a RPPA algorithm $\A_1$ is  \emph{right-consistent} ($\A_1\in\mathbf{C_R}$) if it never offers a price lower than the last accepted one~\cite{2017-WWW-Drutsa}.
The algorithm PRRFES was further modified by the transformation $\pre$ to obtain the one that never decreases offered prices and has  a tight strategic regret bound of $\Theta(\log\log T)$  as well~\cite{2018-ICML-Drutsa}.

%The following  notations are adopted in the RPPA setting~\cite{2014-NIPS-Mohri,2017-WWW-Drutsa,2018-ICML-Drutsa}. 
The workflow of a RPPA algorithm $\A_1$ is usually described by a labeled binary tree $\T(\A_1)$~\cite{2014-NIPS-Mohri,2017-WWW-Drutsa,2018-ICML-Drutsa}: initialize the tracking node $\n$ to the root $\treeroot(\T(\A_1))$;
in each round, the label $\p(\n)$  is offered as a price; if it is accepted (rejected), move the tracking node to the right child $\n\!:=\!\rt(\n)$ (the left child $\n\!:=\!\lf(\n)$, resp.); and go to the next round.
The left (right) subtrees rooted at the node $\lf(\n)$ ($\rt(\n)$, resp.) are denoted by $\LL(\n)$ ($\RR(\n)$, resp.). 
When  trees $\T_1$ and $\T_2$ have the same node labeling, we write $\T_1 \!\cong\! \T_2$.

%The following  notations are adopted in the RPPA setting~\cite{2014-NIPS-Mohri,2017-WWW-Drutsa,2018-ICML-Drutsa}. A RPPA algorithm $\A_1$ is associated with an infinite complete binary tree $\T(\A_1)$: each node $\n\!\in\!\T(\A_1)$ is labeled with the price $\p(\n)$ offered by $\A_1$. 
%%We denote the node's depth + 1 by $t^{\n}$.
%The right and left children of $\n$ are denoted by $\rt(\n)$ and $\lf(\n)$, respectively. The left (right) subtrees rooted at the node $\lf(\n)$ ($\rt(\n)$, resp.) are denoted by $\LL(\n)$ ($\RR(\n)$, resp.). 
%%The operators $\lf(\cdot)$ and $\rt(\cdot)$ sequentially applied $s$ times to a node $\n$ are denoted by  $\lf^s(\n)$ and $\rt^s(\n)$ respectively, $s\in\mathbb{N}$.
%The root node of a tree $\T$ is denoted by $\treeroot(\T)$.
%When two trees $\T_1$ and $\T_2$ have the same node labeling, we write $\T_1 \cong \T_2$.
%Thus, the algorithm's work flow is described as follows: it sets the tracking node $\n$ to the root $\treeroot(\T(\A_1))$;
%at each round, the price $\p(\n)$ is offered; if it is accepted (rejected), the tracking node is moved to the right child $\n:=\rt(\n)$ (the left child $\n:=\lf(\n)$, resp.); and the algorithm goes to the next round.
%%The pseudo-code of this process is presented in Alg.~\ref{alg_A}.
%%For a node $\n\in\T(\A)$, $t^{\n}$ equals to the round at which the price of this node is offered. 
%%{\bf TODO: may be say smth about the strategy}
%%Each node $\n\in\T(\A)$ uniquely determines the buyer decisions up to the round number equal to the depth of the parent node.

\begin{definition}
	\label{def_PenalNodeSeq}
	For a RPPA algorithm $\A_1\!\in\!\ASet^{\RPPA}$, nodes $\n_1,...,\n_r\in\T(\A_1)$ are said to be a \emph{($r$-length) penalization sequence} if
	$ \n_{i+1} \!=\! \lf(\n_i)$, $\p(\n_{i+1}) \!=\! \p(\n_i)$, and  $\RR(\n_{i+1}) \!\cong\! \RR(\n_{i}),   i\!=\!1,..,r\!-\!1$.
\end{definition}
%
%{\bf TODO: Say that we cannot directly apply optimal RPPA algorithms because all the theoretical guarantees are obtained under assumption of complete information (due to the absence of rival bidders) and getting the strategies in advance, in particular Prop.2 in 2017 and Prop.X in 2018 - maybe say it in Sec.2.4 ? }
%Our research goal comprises closing of that open research question (TODO: question on how to adapt 1-buyer algorithm to the multi-bidder case?).
%

%\subsection{Overview of our results}
%\label{subsec_ResultOverview}
{\bf Overview of our results.}
We cannot directly apply the optimal RPPA algorithms~\cite{2017-WWW-Drutsa,2018-ICML-Drutsa}, because our bidders have incomplete information in the game, while the proofs of optimality of these algorithms strongly rely on complete information.  
This completely different information structure of the multi-buyer game results in very complicated bidder behavior even in the absence of reserve prices~\cite{bikhchandani1988reputation}.
Hence, it is challenging to find,  in the multi-buyer case, a pricing algorithm that has regret upper bound of the same asymptotic behavior as the best one in the $1$-buyer RPPA setting.
Our research goal comprises closing of this research question on the existence of such algorithms. 

First, we propose a novel technique to transform a RPPA algorithm to our setup that is based on cyclic elimination of all bidders except one by means of high enough prices (Sec.~\ref{sec_divClass}).
Separate playing with each buyer removes his uncertainty about the outcome of a current round; and, despite remaining uncertainty about future rounds, this is enough to construct a tool to locate his valuation (Prop.~\ref{prop_reject_bound_divalg}).
Second,  we transform PRRFES in this way and show that its regret is affected by two learning processes: the one learns bidder valuations and  the other learns which bidders have the maximal valuation (Sec.~\ref{sec_divPRRFES}).
The former learning is controlled by the design of the source PRRFES, while the latter one is achieved by a special stopping rule that excludes bidders from suspected ones. A proper combination of parameters for the source pricing and the stopping rule provides an algorithm with strategic regret in $O(\log\log T)$, see Th.~\ref{th_divPRRFES_sreg_uppbound}.

% We show that this div-transformation of right-consistent algorithm with proper penalization allow to locate and learn the valuation of a strategic buyer even in our case of uncertainty about the future.
%We apply this transformation to an optimal RPPA algorithm and show that the resulting algorithm has strategic regret upper bound of the same asymptotic $O(\log\log T)$ 

%\subsection{Related Work}
%\label{subsec_RelWork}
{\bf Related work.}
Several studies maximized revenue of auctions in an offline/batch learning fashion:  either via 
estimating or fitting of distributions of buyer valuations/bids to set  reserve prices~\cite{2013-IJCAI-He,2014-ECRA-Sun,2016-WWW-Paes}, 
or via direct learning of reserve prices~\cite{2014-ICML-Mohri,2015-UAI-Mohri,2016-WWW-Rudolph,2017-NIPS-Medina}.
In contrast to them, we set prices in repeated auctions by  an online deterministic learning approach.
Revenue optimization for repeated
auctions was mainly concentrated on algorithmic reserve prices, that are updated in online way over time, and was also known as dynamic pricing~\cite{fudenberg2006behavior,2015-SORMS-den-Boer}.
Dynamic pricing was considered: 
under game-theoretic view~\cite{leme2012sequential,2015-EC-Chen,2016-EC-Balseiro,2016-EC-Ashlagi,mirrokni2018optimal};
from the bidder side~\cite{2011-ECOMexch-Iyer,2016-JMLR-Weed,2016-ICML-Heidari,2017-NIPS-Baltaoglu}; 
in experimental studies~\cite{list1999price,2012-RIO-Carare,2014-KDD-Yuan};
as bandit problems~\cite{2011-COLT-Amin,2015-NIPS-Lin,cesa2018dynamic}; 
and from other aspects~\cite{2016-EC-Roughgarden,2016-NIPS-Feldman,2016-SODA-Chawla,2018-IJGT-Hummel}.
Repeated auctions with a contextual information about the good in a round were considered in~\cite{2014-NIPS-Amin,2016-EC-Cohen,2018-NIPS-Mao,2018-NIPS-Leme}. 
The studies~\cite{1993-JET-Schmidt,hart1988contract,2015-SODA-Devanur,2017-EC-Immorlica,2018-ArXiV-Vanunts} elaborated on setups of repeated posted-price auctions with a strategic buyer holding a fixed valuation, but maximized expected revenue for a given prior distribution of valuations, while we optimize regret w.r.t.\ worst-case valuations without knowing their distribution.

There are studies on reserve price optimization in repeated second-price auctions, but they considered scenarios different to ours.
Non-strategic bidders are considered in~\cite{2013-SODA-Cesa-Bianchi}.
Kanoria et al.~\cite{2014-SSRN-Kanoria} studied strategic buyers (similarly to our work), but maximized expected revenue w.r.t.\ a prior distribution of valuations.
Our setup can be considered as a special case of repeated Vickrey auctions in~\cite{2018-NIPS-Huang}, but their regret upper bound is $O(T^\alpha)$ in $T$ and holds only when selling several goods in a round.
However, the most relevant works to ours
are~\cite{2013-NIPS-Amin,2014-NIPS-Mohri,2017-WWW-Drutsa,2018-ICML-Drutsa}, where our strategic setup with fixed private valuation is considered, but for the case of one bidder, $M=1$. The most important results of these works are discussed above in this section (see ``Background on pricing algorithms").
% Sec.~\ref{subsec_CurrentKnow}. 

%=======
% TODO:
%FOR OTHER PLACES in the text:
%
%\cite{peters2006internet} not our case at all | JUST CITE | Use PBE!
%
%

%\textbf{Personal/individual RP:}
%2014-SSRN-Kanoria   --- Described SP with anonym RP?
%hartline2009simple
%fu2012ad
%2015-NIPS-Morgenstern
%2016-WWW-Paes --- MUST MUST (no dynamic)
%2016-EC-Roughgarden  <--- dynamic, generally VCG, bidders directly reports their values, while in our case they may lie and behave strategically (IN ANY WAY: mechanism without strateginess is quite similar to ours).

%\textbf{complex behavior}
%\cite{bikhchandani1988reputation} MUST USE as demonstration of a very complex behavior in repeat SP auctions (even in the absence of RP) | Note: valuations are static

\section{Dividing algorithms and $\diva$-transformation}
\label{sec_divClass}
{\bf Barrage pricing.}
In our setting, a pricing algorithm is able to set personal (individual) reserve prices to each bidder and is able hence to ``eliminate" particular bidders from particular rounds.
Namely, in a round $t$, an algorithm can set a reserve price $p^{\mathrm{bar}}$ s.t.\ a strategic bidder $m$, independently of his valuation, will never accept $p^{\mathrm{bar}}$, i.e., will never bid no lower than this price; such a price is referred to as a \emph{barrage reserve price}.
From here on we use $p^{\mathrm{bar}}=1/(1-\g_0), \g_0\in(0,1)$: accepting it once will result in a negative surplus for a buyer with discount $\g_i\le\g_0$.
We use the phrase ``the bidder $m$ is \emph{eliminated}\footnote{Note that, (a)~formally, all bidders participate in all rounds (see Sec.~\ref{sec_Fwk}) and (b), if a bidder is not eliminated, it does not mean that he is in $\Bset_t$ (he may bid below his reserve price which can be a non-barrage one). So, the word ``elimination" is purposely associate with barrage pricing in order to refer to this case.} from participation in the round $t$" to describe this case.
%If a bidder is not eliminated in this way, then we sometimes refer to him as an \emph{admitted participant}.

{\bf Dividing algorithms.}
In this subsection, we introduce a subclass of the algorithms $\ASet_{M}$ that is denoted by $\ASet^{\diva}_{M}\subset\ASet_{M}$ and is referred to as \emph{the class of dividing algorithms}\ (stands for lat. ``Divide et impera").
%``Divide and rule").
%, ).
A dividing algorithm $\A\in \ASet^{\diva}_{M}$ works in \emph{periods} and tracks a feasible set of \emph{suspected bidders} $\Sset$ aimed to find the bidder (or bidders) with the maximal valuation $\bv$.
Namely, it starts with all bidders $\Sset_1:=\Bset$ at the first period which lasts $M$ rounds. 
In each period $i\in\mathbb{N}$, the algorithm iterates over the currently suspected bidders $\Sset_i$: in a current round, it picks up $m\in\Sset_i$, sets a non-barrage reserve price to the bidder $m$, sets a barrage reserve price to all other bidders $\Bset^{-m}$, and goes to the next round within the period by picking up the next buyer from $\Sset_i$.
Thus, the algorithm meaningfully interacts with only one bidder in each round through elimination of all other bidders by means of barrage pricing. 
%So, the period $i$ consists of $\tau_i:=|\Sset_i|$ rounds.
After the $i$-th period, the algorithm $\A$ identifies somehow which bidders from $\Sset_i$ should be left as suspected ones  in the next period (i.e., be included in the set $\Sset_{i+1}$).

When the game has been played with the dividing algorithm $\A$, one can split all the rounds into $I$ periods: $\{1,\ldots,T\} = \cup_{i=1}^{I}\TT_i$. 
Each period $i<I$ consists of $|\TT_i|=|\Sset_i|$ rounds (the last one of $|\TT_I|\le|\Sset_I|$).
Let  $t^m_i\in\TT_i$ denote the round of a period $i$ in which a bidder $m$ is not eliminated by the seller algorithm (i.e., receives a non-barrage reserve price).
Thus, $\II^m:=\{t^m_1,...,t^m_{I^m}\}$ are all such rounds of the bidder $m$ and 
$I^m=|\II^m|$ is referred to as the \emph{subhorizon} of the bidder $m$  (the number of periods where he participates).
Note that (a) $I^m$ and $\II^m$ depend on the bids $\bb_{1:T}$ of all buyers $\Bset$; (b) the following identities hold: $\{1,\ldots,T\} = \cup_{m=1}^M\II^m$ and $\TT_i=\{t^m_i \mid m\in\Bset \,\hbox{s.t.}\, I^m \ge i \}$.

So, in a round $t^m_i$, the algorithm $\A$ eliminates the bidders $\Bset^{-m}$ (i.e., sets the reserves $p^{m'}_{t^m_i}\! =\!p^{\mathrm{bar}}$ $ \fa {m'}\!\in\!\Bset^{-m}$), while the reserve price $p^m_{t^m_{i}}$ set for the buyer $m$ is determined only by his bids during the previous rounds $\{t^m_{1},..., t^m_{i-1}\}$ where he has not been eliminated: i.e., $p^m_{t^m_{i}} \!=\! p^m(b^m_{t^m_1},...,b^m_{t^m_{i-1}})$.
Hence, the algorithm $\A$'s interaction with the bidder $m$ in the rounds $\II^m$ can be encoded by a $1$-buyer algorithm from $\ASet_{1}$, which sets prices in the rounds $\{t^m_i\}_{i=1}^{I^m}$ instead of $\{i\}_{i=1}^{I^m}$. 
We denote this algorithm by $\A^m$ and refer to it as \emph{the subalgorithm} of $\A$ against the buyer $m$.
Let $\Reg^m(\II^m,\A^m,v^m,b^m_{1:T})\!:=\!\sum_{i=1}^{I^m}(v^m \!-\! a^m_{t^m_i}p^m_{t^m_i} )$ be the regret of the subalgorithm $\A^m$ for given bids $b^m_{1:T}$ of the buyer $m\in\Bset$ in the rounds $\II^m$.
The lemma holds (the trivial proof is in App.~A.1.1).
\begin{lemma}
	\label{lemma_divalg_reg_decompose}
	Let $\A\in\ASet^{\diva}_{M}$ be a dividing algorithm, $\A^m\in\ASet_{1}, m\in\Bset$, be its subalgorithms (as described above), and $\obb_{1:T}=\obb_{1:T}(T,\A,\vb,\gb,\betab)$ be optimal bids of the strategic buyers $\Bset$. Then, for any $\vb\in[0,1]^M$, $ \gb\in(0,1]^M,$ and $\betab$,  the strategic regret of $\A$ can be decomposed into two parts  
	$\SReg(T,\A,\vb,\gb,\betab) = \SReg^{\ind}(T,\A,\vb,\gb,\betab) + \SReg^{\dev}(T,\A,\vb,\gb,\betab)$,
	where 
%	\vspace{-0.1cm}
%	\begin{equation}
%	\label{lemma_divalg_reg_decompose_eq_2}
	$\SReg^{\ind}(T,\A,\vb,\gb,\betab):=\sum_{m\in\Bset} \Reg^m(\II^m,\A^m,v^m,\ob^m_{1:T})$
%	\vspace{-0.1cm}
%	\end{equation}
	is the individual part of the regret and
%	\vspace{-0.1cm}
%	\begin{equation}
%	\label{lemma_divalg_reg_decompose_eq_3}
	$\SReg^{\dev}(T,\A,\vb,\gb,\betab) := \sum_{m\in\Bset} I^m (\bv - v^m)$
%	\vspace{-0.1cm}
%	\end{equation}
	is the deviation part  of the regret.
\end{lemma}
Informally, this lemma states that the regret consists of the individual regrets against each buyer $m$ in his rounds $\II^m$
%$\II^m$ 
and the deviation of the buyer valuations $\vb$ from the maximal one $\bv$.
So, we see a clear intuition: a good algorithm should (1) \emph{learn the valuations $\vb$ of the buyers} (minimizing individual regrets) and (2) \emph{learn which buyers have the highest valuation $\bv$} (minimizing the deviation regret).

{\bf $\diva$-transformation.}
Let $\A_1\in\ASet^{\RPPA}$ be a $1$-buyer RPPA algorithm. An algorithm $\diva_M(\A_1,\sr)$ is said to be \emph{a $\diva$-transformation of $\A_1$} with \emph{a stopping rule} $\sr:\Bset\times\T(\A_1)^M\to\mathtt{bool}$ when it is a dividing algorithm from $\ASet^{\diva}_{M}$ s.t.\ its subalgorithms $\A^m$ are $\A_1$ and the stopping rule $\sr$  determines which bidders are not suspected ones in $\Sset_{i+1}$ after a period $i$. Namely, first, the algorithm $\diva_M(\A_1,\sr)$ tracks the state of each buyer $m\in\Bset$ in the tree $\T(\A_1)$ of the RPPA algorithm $\A_1$ (see Sec.~\ref{sec_Fwk}) by means of a personal (individual) feasible node. 
For each period $i$ and for each round $t^m_i \in\TT_i$, the current state (i.e., the history of previous actions) of the buyer $m$ is encoded by the tracking node $\n_i^m\in\T(\A)$; in particular, in the round $t^m_i$, he receives the reserve price $\p(\n_i^m)$ of this node  $\n_i^m$  (the other bidders $\Bset^{-m}$ get a barrage reserve price $p^{\mathrm{bar}}$).
If a buyer $m$ is not more suspected in a period $i>I^m$ (i.e., $m\not\in\Sset_{i}$), we formally set $\n_{i}^m:=\n_{I^m+1}^m$.
Second, after a period $i$, the stopping decision is based on the past buyer binary actions that are coded by means of the nodes $\{\n_{i+1}^m\}_{m=1}^M$ in the binary tree $\T(\A_1)$: if the stopping rule $\sr(m', \{\n_{i+1}^m\}_{m=1}^M)$ is $\mathtt{true}$, then the buyer $m'\not\in\Sset_{i+1}$.
The pseudo-code of the $\diva$-transformation of a RPPA algorithm is  in Appendix~B.1.

%\begin{algorithm}[t]
%	\small
%	\caption{Pseudo-code of a $\diva$-transformation $\diva_M(\A_1,\sr)$ of a RPPA algorithm $\A_1\in\ASet^{\RPPA}$.}
%	\label{alg_div_transform}
%	\begin{algorithmic}[1]
%%		
%		\STATE {{\bfseries Input:}$M\in\mathbb{N}$,  $\A_1\in\ASet^{\RPPA},$ $\sr:\Bset\times\T(\A_1)^M\to \mathtt{bool}$}
%		\STATE {{\bfseries Initialize:} $t:=1$, $\Sset :=\Bset,$ $\n[\:] := \{\treeroot(\T(\A_1))\}_{m=1}^M,$}
%		\WHILE{$t\le T$}
%		    \FORALL{$m\in\Sset$}		    
%		      \STATE {Set the price $\p(\n[m])$ as reserve to the buyer $m$}
%			  \STATE {Set the price $p^{\mathrm{bar}}$  as reserve to the buyers from $\Bset^{-m}$}
%			  \STATE {$\bb[\:] \gets $ get bids from the buyers $\Bset$}
%		      \IF{$\bb[m]\ge \p(\n[m])$}
%		        \STATE {Allocate $t$-th good to the buyer $m$ for the price $\p(\n[m])$}
%			    \STATE {$\n[m] := \rt(\n[m])$}
%			  \ELSE
%			    \STATE {$\n[m] := \lf(\n[m])$}
%			  \ENDIF
%			  \STATE {$t := t+1$}
%		      \IF{ $t> T$}
%			   \STATE {{\bf break} } 
%			  \ENDIF
%		    \ENDFOR
%		    \STATE {$\Sset^{\mathrm{old}}:=\Sset$}
%		    \FORALL{$m\in\Sset^{\mathrm{old}}$}		    
%		      \IF{  $\sr(m,\n[\:])$}
%     		    \STATE {$\Sset:=\Sset\setminus\{m\}$}
%		      \ENDIF
%		    \ENDFOR
%		\ENDWHILE	
%	\end{algorithmic}
%\end{algorithm}

For a RPPA right-consistent  algorithm  $\A_1\!\in\!\mathbf{C_R}$  with penalization rounds, let $\langle\A_1\rangle$ denote the~transformation of $\A_1$ s.t.\ it is equal to $\A_1$, but each penalization sequence of nodes $\{\n_j\}_{j=1}^r\!\subset\!\T(\A_1), r\ge 2,$  (see Def.~\ref{def_PenalNodeSeq}) is \emph{reinforced} in the following way: all the prices in the nodes $\{\n_{j}\}\cup\RR(\n_{j}),  j=2,...,r,$ are replaced by $1$ (the maximal valuation domain value); the sequence and the rounds are then referred to as \emph{reinforced penalization} ones.
After this, a strategic buyer will  certainly either accept the price at the node $\n_1$, or reject the prices in all the nodes $\{\n_j\}_{j=1}^r$ even in the case of uncertainty about the future. 
Let $\delta_{\n}^{l}:= \p(\n) - \inf_{\m\in\LL(\n)}\p(\m)$ be the left increment~\cite{2014-NIPS-Mohri,2017-WWW-Drutsa} of a node $\n\in\T(\A_1)$.

In order to obtain upper bounds on strategic regret, it is important to have a tool that allows to locate the valuation of a strategic bidder.
Such a tool can be obtained for $\diva$-transformed right-consistent RPPA algorithms with reinforced penalization rounds based on the following proposition, which  is an analogue of \cite[Prop.2]{2017-WWW-Drutsa} in our case with buyer uncertainty about the future. 

\begin{proposition}
	\label{prop_reject_bound_divalg}
	Let $\g_m\in(0,1)$, $\A_1\in\ASet^{\RPPA}\cap\mathbf{C_R}$ be a RPPA right-consistent pricing algorithm, $\n\!\in\!\T(\A_1)$ be a starting node in a $r$-length penalization sequence (see Def.~\ref{def_PenalNodeSeq}), $r > \log_{\g_m}(1-\g_m)$, $\sr\!:\!\Bset\!\times\!\T(\A_1)^M\!\to\! \mathtt{bool}$ be a stopping rule, and the $\diva$-transformation $\diva_M(\langle\A_1\rangle,\sr)$ be used by the seller for setting reserve prices. If, in a round, the node $\n$ is reached and the price $\p(\n)$ is rejected by a strategic buyer $m\!\in\!\Bset$ (i.e., he bids lower than $\p(\n)$), then the following inequality on $v^m\!$ 
%	on his valuation $v^m$ 
	holds:
%	\vspace{-0.1cm}
	\begin{equation}
	\label{prop_reject_bound_divalg_eq_1}
	v^m - \p(\n) < \zeta_{r,\g_m} \delta_{\n}^{l}, \qquad \hbox{where} \qquad \zeta_{r,\g} := {\g^r}/({1-\g-\g^r}). 
%	\vspace{-0.1cm}
	\end{equation}
\end{proposition}

\begin{proof}[Proof sketch]	The full proof is in App.A.1.2.
	Let  $t$ be the round in which the bidder $m$ reaches the node $\n$ and rejects his reserve price $p^{m}_t\!=\!\p(\n)$.
	In particular, it is the round where he is the non-eliminated buyer and $t\! = \!t^m_i\!\in\!\TT_i$ for some period $i$.
	Since the buyers are divided and $\A_1\!\in\!\ASet^{\RPPA}$, w.l.o.g., any strategy can be treated as a map to binary decisions $\{0,1\}$.
	Let 
	$\osig$ be the optimal strategy used by the buyer $m$;
	$h^m_{t;a}$ be the continuation of the current history $h^m_t$ by a binary decision $a^m_t\! = \!a$, while 
	$\hat{\sig}_a$ denote an optimal strategy among all possible strategies in which the binary buyer decision $a^m_t$ is $a\in\{0,1\}$;
	and $S^m_t(\sig):=\Sur_{t:T}(\A,\g_m,v^m,h^m_t,\beta^m,\sig)$ be the future expected surplus when following a strategy $\sig\!\in\!\StrartSet_T$.
	Rejection of the price $p^{m}_t$ when following the optimal strategy $\osig$ easily implies: $S^m_t(\hat{\sig}_1) \le S^m_t(\hat{\sig}_0)$.
	Let us bound each side of this inequality. First, 
%	\vspace{-0.4cm}
	\begin{equation}
	\label{prop_reject_bound_divalg_proof_eq_1}
%	\begin{split}
%	&
	S^m_t(\hat{\sig}_1)  = \g_m^{t-1}(v^m-\p(\n))
%	+ \\&
    + \Sur_{t+1:T}(\A,\g_m,v^m,h^m_{t;1},\beta^m,\hat{\sig}_1)
	\ge\g_m^{t-1}(v^m-\p(\n)),
%	\end{split}
%	\vspace{-0.3cm}
	\end{equation}
	where we used the facts (i)  that if the bidder accepts the price $\p(\n)$, then he necessarily gets the good since all other bidders $\Bset^{-m}$ are eliminated by a barrage price in this round~$t$; and (ii) that the expected surplus in rounds $s\ge t+1$ is at least non-negative, because the subalgorithm $\A_1\in\mathbf{C_R}$ is right-consistent. Second,
	$
	S^m_t(\hat{\sig}_0) =
	\Sur_{t^m_{i+r}:T}(\A,\g_m,v^m,h^m_{t;0},\beta^m,\hat{\sig}_0)
	< \frac{\g_m^{t+r-1}}{1-\g_m}(v^m - \p(\n) + \delta_{\n}^{l}),
	$
	where we (i) used the fact that if the bidder rejects the price $p^{m}_t$, then the future rounds $\{t^m_{i+j}\}_{j=1}^{r-1}$ will be reinforced penalization ones (the strategic bidder will reject in all of them); and (ii) upper bounded the surplus in  remaining rounds by assuming that only this bidder will get remaining goods for the lowest reserve price from the left subtree $\LL(\n)$.
	We unite these bounds on $S^m_t(\hat{\sig}_a)$ and get
%	\vspace{-0.1cm}
%	\begin{equation}
%	\label{prop_reject_bound_divalg_proof_eq_3}
	$(v^m - \p(\n))\left( 1-\g_m - \g_m^r\right)    <  \g_m^{r} \delta_{\n}^{l},$
%	\vspace{-0.1cm}
%	\end{equation}	
%	dividing by $\g_m^{t-1}/(1-\g_m)$. 
	what implies Eq.~(\ref{prop_reject_bound_divalg_eq_1}), since $r > \log_{\g_m}(1-\g_m)$. 
\end{proof}

We emphasize that \emph{the dividing structure of the algorithm is crucially exploited} in the proof of Prop.~\ref{prop_reject_bound_divalg}. Namely, the fact that all other bidders $\Bset^{-m}$ are eliminated by a barrage price in the round $t$ is used 
(a) to guarantee obtaining of the good at price $\p(\n)$ by the buyer $m$  and  
(b) to lower bound thus the future surplus $S^m_t(\hat{\sig}_1)$ in the case of acceptance in Eq.~(\ref{prop_reject_bound_divalg_proof_eq_1}). If we dealt with a non-dividing algorithm, then another bidder might win the good or make the payment of the bidder $m$ higher than his reserve price $\p(\n)$; in both cases, $S^m_t(\hat{\sig}_1)$ could only be lower bounded by $0$ in a general situation, what would result in an useless inequality instead of Eq.~(\ref{prop_reject_bound_divalg_eq_1}).

For a right-consistent algorithm $\A_1\in\mathbf{C_R}$, the increment $\delta_{\n}^{l}$ is bounded by the difference between the current node's price $\p(\n)$ and the last accepted price $q$ by the buyer $m$  before reaching this node.
Hence, the Prop.~\ref{prop_reject_bound_divalg} provides us with a tool  to locate the valuation $v^m$ despite the strategic buyer does not myopically report its position (similar to~\cite[Prop.2]{2017-WWW-Drutsa}). Namely, if the buyer $m$ bids no lower than $\p(\n)$, then $v^m \ge \p(\n)$; if he bids lower than $\p(\n)$, then $q \!\le\! v \!<\! \p(\n) \!+\! \zeta_{r,\g_m} (\p(\n) - q)$ and  the closer an offered price $\p(\n)$ is to the last accepted price $q$ the smaller the location interval of possible valuations $v^m$ (since its length is $(1 +\zeta_{r,\g_m}) (\p(\n) \!-\! q)$).

\section{divPRRFES algorithm}
\label{sec_divPRRFES}
In this section, we will show that we can use an optimal algorithm from the setting of repeated posted-price auctions  to obtain the algorithm for our multi-bidder setting with upper bound on strategic regret with the same asymptotic. Namely, let us $\diva$-transform
%\footnote{In fact, the optimal algorithm prePRRFES~\cite{2018-ICML-Drutsa} can also be considered giving the same result (see Sec.~\ref{TODO}{\bf TODO})}
PRRFES~\cite{2017-WWW-Drutsa}, further denoted as $\A_1$.
%  to get the subalgorithm $\langle\A_1\rangle$ and $\diva$-transform it.

Since a $\diva$-transformation of PRRFES (with penalization reinforcement) individually tracks position of each buyer in the binary tree $\T(\langle\A_1\rangle)$, we adapt the key notations of PRRFES~\cite{2017-WWW-Drutsa} to our case of multiple bidders and periods. 
Against a buyer $m\!\in\!\Bset$, PRRFES $\langle\A_1\rangle$ works in phases initialized by the phase index $l\!:=\!0$, the last accepted price before the current phase $q^m_0\!:=\!0$, and the iteration parameter $\e_0\!:=\!1/2$; at each phase $l\!\in\!\mathbb{Z}_{+}$, it sequentially offers prices $p^m_{l,k}\!:=\! q^m_{l} \!+\! k\e_{l}, k\!\in\!\mathbb{N}$ (\emph{exploration rounds}), with $\e_{l}\!=\!2^{-2^l}$;
if a price $p^m_{l,k}$  is rejected,  setting $K^{m}_{l}\!:=\!k \!-\! 1\!\ge\! 0$,
(1)~it offers the price $1$ for $r\!-\!1$ \emph{reinforced penalization rounds} (if one of them is accepted, $1$ will be offered in all remaining rounds), 
(2)~it offers the price $p^m_{l,K^{m}_{l}}$ for $g(l)$ \emph{exploitation rounds}, and 
(3)~PRRFES goes to the next phase by setting $q^{m}_{l+1}\!:=\!p^m_{l,K^{m}_{l}}$ and $l:=l+1$.
Individual tracking of bidders  by the $\diva$-transformed PRRFES implies that different buyers can be in different phases in the same period~$i$. 
Hence, let $l^m_i$ denote the current phase of a buyer $m\!\in\!\Bset$ in the round $t^m_i$ of a period~$i\!\le\! I^m$, 
and let $l_{i}^m\!:=\!l_{I^m+1}^m$ 
in all subsequent periods $i\!>\!I^m$ (when the buyer $m$ is no more suspected).
In particular, $q^m_{l^m_i}$ is the last accepted price by the buyer $m$ before the phase $l^m_i$ in the period $i$.
We rely on the decomposition from Lemma~\ref{lemma_divalg_reg_decompose} in order to bound the strategic regret of a $\diva$-transformed PRRFES.

%[
%
%Despite we separately play with each buyer (due to div-transformation) and remove uncertainty for him in current round, he still has uncertainty about future rounds: actions of the other bidders affect how many times seller will play with him. {\bf ---- this sentence good for intro }
%
%]

{\bf Upper bound for individual regrets.}
Before specifying a particular stopping rule, let us obtain an upper bound on individual strategic regret $\Reg^m(\II^m,\langle\A_1\rangle, v^m,\ob^m_{1:T}), m\!\in\!\Bset$. This regret is not equal to $\SReg(I^m\!,\langle\A_1\rangle,(v^m\!),(\g_m))$ since, in the latter case, the $1$-bidder game does not depend on behavior of the other bidders $\Bset^{-m}$ (while, in the former case, does). 
In other words, the rounds $\II^m\!\!=\!\!\{t^m_i\}_{i\!=\!1}^{I^m}$ do not constitute the $I^m$-round $1$-buyer game of the RPPA setting  considered in \cite{2013-NIPS-Amin,2017-WWW-Drutsa}, because the subhorizon $I^m$ and exact rounds $\II^m$ (they determine  the used discount factors: $\g_m^{t-1}, t\in\II^m$) are unknown in advance and depend on actions of the other bidders. 
%Hence, upper bounds for $\SReg(I^m,\A^m,\{v^m\}, \{\g_m\})$ could not be directly used to upper bound $\Reg^m(\II^m,\A^m,v^m,\ob^m_{1:T})$. 
Hence, this does not allow to straightforwardly utilize the result on the strategic regret for PRRFES  proved in~\cite[Th.5]{2017-WWW-Drutsa} for the setting of RPPA. So, we have to prove the bound $O(\log\log T)$ for our case  with buyer uncertainty about the future.
Let us introduce the notation:  
%\vspace{-0.2cm}
%\begin{equation}
%TODO: maybe return!!!
%\label{eq_penal_round_lowbound_on_gamma}
$r_{\g} := \big\lceil \log_{\g}\big((1-\g)/2\big) \big\rceil  \fa \g\in(0,1)$.
%\vspace{-0.2cm}
%\end{equation}
%{\bf TODO: need we introduce the notation $r_g$?}

\begin{lemma}
	\label{lemma_divPRRFES_indreg_uppbound}
	Let $\g_0\!\in\!(0,1)$, $\A_1$ be the PRRFES algorithm with $r \!\ge\! r_{\g_0} $
	%	:= \big\lceil \log_{\g_0}\big((1-\g_0)/2\big) \big\rceil$ 
	and  the exploitation rate $g(l) \!=\! 2^{2^{l}}, l\in\mathbb{Z}_+$, and $\sr\!:\!\Bset\!\times\!\T(\A_1)^M\!\to\! \mathtt{bool}$ be a stopping rule.
	Then, 
%	for any valuation $v^m\in [0,1]$, 
	if $I^m\!\ge\!2$, the individual regret of the $\diva$-transformed PRRFES  $\diva_M(\langle\A_1\rangle,\sr)$ against the buyer $m\!\in\!\Bset$ is upper bounded: 
%	\vspace{-0.2cm}
	\begin{equation}
%	\begin{split}
	\label{lemma_divPRRFES_indreg_uppbound_eq1} 
%	&
	\Reg^m(\II^m,\langle\A_1\rangle,v^m,\ob^m_{1:T})  \le
%	 \\&\le 
	 (rv^m+4)(\log_2\log_2 I^m + 2) \quad \fa\g_m\in(0,\g_0] \fa v^m\in [0,1],
%	\end{split}
%	\vspace{-0.4cm}
	\end{equation}
	where 
%	$I^m\!\!=\!\!|\II^m|$ 
%	is the subhorizon of the buyer $m$ in this game
%	and 
	$\obb_{1:T}=\obb_{1:T}(T,\diva_M(\langle\A_1\rangle,\sr),\vb,\gb,\betab)$ are optimal bids of the strategic buyers $\Bset$.
\end{lemma}
\vspace{-0.2cm}
\begin{proof}[Proof sketch]
	Decompose the  individual regret over the rounds $\II^m$ into the sum of the phases'~regrets: $\Reg^m\!(\II^m,\!\langle\A_1\rangle,\!v^m,\!\ob^m_{1:T})\!\!=\!\! \sum_{l=0}^{L^m}\!R^m_l$, where $L^m\!\!:=\!\!l^m_{I^m}$ is the number of phases conducted by the algorithm against the  buyer $m$. For  $l\!\in\!\mathbb{Z}_{L^m-1}$:
	$
	R^m_l \!\!=\!\!  \sum_{k=1}^{K^m_l} (v^m\!-\!p^m_{l,k})\! +\! rv^m \!+\! g(l)(v^m\!-\!p^m_{l,K^m_{l}}),
	$
	where the terms correspond to the accepted exploration rounds, the reject-penalization ones, and the exploitation ones.
	PRRFES and each rejected price $p^m_{l,K^m_{l}+1}$ satisfy the conditions of Prop.~\ref{prop_reject_bound_divalg}, what implies $v^m \!-\! p^m_{l,K^m_{l} + 1} \!<\! (p^m_{l,K^m_{l} \!+\! 1} \!-\! p^m_{l,K^m_{l}})\!=\!\e_l$ (since $\zeta_{r,\g_m}\!\le\!1$ for $r\!\!\ge\!\! r_{\g_0}$ and $\g_m\!\!\le\!\! \g_0$). Hence,  $v^m\!\!\in\!\![q^m_{l+1}, \!q^m_{l+1} \!\!+\!\! 2\e_{l})$ (since $q^m_{l+1}\!\! =\!\! p^m_{l,K^m_{l}}$ and PRRFES is right-consistent) and the number of exploration rounds is thus bounded: $K^m_{l+1} \!<\!  2^{2^{l}+1}$.
	All further steps are similar to~\cite[Th.5]{2017-WWW-Drutsa}: 
	$\sum_{k=1}^{K^m_{l}}(v^m\!-\!p^m_{l,k}) \!<\! 2$; 
	for each phase $l$, we get that  $R^m_l \!\!\le\!\!  rv^m+4$;
	and the number of phases $L^m \!\!\le\!\! \log_2\log_2 I^m \!\!+\!\! 1$.
	The full proof is in  Appendix~A.2.1 of Supp. Materials.
\end{proof}
\vspace{-0.2cm}

{\bf Upper bound for deviation regret.}
Prop.~\ref{prop_reject_bound_divalg} provides us with the tool that locates the valuation $v^m$ of a bidder $m\!\in\!\Bset$ at least in the segment $[u^m_i,w^m_i]:=[q^m_{l_i^m}, q^m_{l_i^m} + 2\e_{l_i^m-1}]$ right after a period $i-1$ (see the proof [sketch] of Lemma~\ref{lemma_divPRRFES_indreg_uppbound}), when $r\ge r_{\g_m}$.
%\big\lceil \log_{\g_m}\big((1-\g_m)/2\big) \big\rceil$.
This means: if, after playing a period $i-1$, the upper bound $w^m_i$ of the valuation of a bidder $m\in\Bset$ is lower that the lower bound $u^{\hat m}_i$ of the valuation of another bidder $\hat m\in\Bset^{-m}$, i.e., $w^m_i \!<\! u^{\hat m}_i$, then the bidder $m$ does definitely  have non-maximal valuation (i.e., $v^m\!<\!\bv$) and needs not to be suspected in the period $i$ and subsequent ones. 
%Hence, based on this observation, one can derive the following stopping rule.
For given parameters $r$ and $g(\cdot)$ of the PRRFES algorithm $\A_1$, any state $\n\!\in\!\T(\A_1)$ of the algorithm can be mapped to the current phase $l(\n)$ and the last accepted price $q(\n)$ before the phase $l(\n)$. 
Thus, we define the stopping rule: $\sr_{\A_1}(m, \!\{\n^m\}_{m=1}^M)\!:=\!\rho(m, \! \{l(\n^m)\}_{m=1}^M,  \!\{q(\n^m)\}_{m=1}^M)$, where
%\vspace{-0.1cm}
\begin{equation}
\label{eq_def_sr_of_PRRFES} 
\rho(m, \lb, \qb):=
\exists \hat m\in\Bset^{-m}:  q^m + 2\e_{l^m-1}<q^{\hat m} \qquad \fa \lb\in\mathbb{Z}^M_+ \fa \qb\in\mathbb{R}^M_+.
%\vspace{-0.1cm}
\end{equation}
The $\diva$-transformation $\diva_M(\langle\A_1\rangle,\sr_{\A_1})$  of the  PRRFES algorithm   $\A_1$ with the stopping rule $\sr_{\A_1}$
%\footnote{Note that our stopping rule $\sr_{\A_1}$ uses the current values $\{l_i^m\}_{m=1}^M$ \&  $\{q^m_{l_i^m}\}_{m=1}^M$ known before a period $i$ (since the tracking nodes have been moved to $\{\n_i^m\}_{m=1}^M$ after playing the period $i-1$) and should be applied with the subalgorithm PRRFES that have the same parameters $r$ and $g(\cdot)$.}
defined in Eq.~(\ref{eq_def_sr_of_PRRFES}) is  referred to as \emph{the dividing Penalized Reject-Revising Fast Exploiting Search} (\emph{divPRRFES}). 
The pseudo-code of divPRRFES is presented in Appendix~B.2 of Supp.Materials. 

\begin{lemma}
	\label{lemma_divPRRFES_devreg_uppbound}
	Let $\g_0\in(0,1)$, the discounts $\gb\in(0,\g_0]^M$,  and the seller uses the divPRRFES pricing algorithm $\diva_M(\langle\A_1\rangle,\sr_{\A_1})$
	with the number of penalization rounds  $r \ge r_{\g_0}$,
	%	 := \big\lceil \log_{\g_0}\big((1-\g_0)/2\big) \big\rceil$, 
	with the exploitation rate $g(l) = 2^{2^{l}}, l\in\mathbb{Z}_+$, and with the stopping rule  $\sr_{\A_1}$ defined in Eq.~(\ref{eq_def_sr_of_PRRFES}).
	Then, for a bidder $m\in\Bset$ with non-maximal valuation, i.e.,   $v^m< \bv$,  his subhorizon $I^m$ is  bounded: 
%	\vspace{-0.1cm}
	%TODO: rechech that this bound is for I^m, not for (I^m-1) - see the proof with the text "the rounds before reaching the period"
	\begin{equation}
	\label{lemma_divPRRFES_devreg_uppbound_eq1} 
	I^m\le  {24}({\bv - v^m})^{-1} + r \big(1 + \log_2\log_2({4}({\bv - v^m})^{-1})\big)<({24+5r})({\bv - v^m})^{-1}.
%	\vspace{-0.2cm}
	\end{equation}    
\end{lemma}
\begin{proof}[Proof sketch]
	Let $\bm$ be a buyer with the maximal valuation $\bv$. Note that, in any period $j=1,..,I^m$, the location intervals  $[q_{l_j^m}^m, q_{l_j^m}^m + 2\e_{l_j^m-1}]$ and $[q_{l_j^{\bm}}^{\bm}, q_{l_j^{\bm}}^{\bm} + 2\e_{l_j^{\bm}-1}]$ must intersect (otherwise, the stopping rule $\sr_{\A_1}$ has eliminated the buyer $m$ before the period $j$, and, hence, $j>I^m$). 
	In particular, in the period $I^m$,  $\e_{L(m',m)}\ge(\bv - v^m)/{4}$ holds for either  $m'=m$ or (not exclusively) $m'=\bm$, where $L(m',m):=l^{m'}_{I^m}$.
	From the definition of the iteration parameter $\e_l$, i.e. $\log_2\e_l = -2^{l}$, one can obtain the bound on one of the phases:
%	\vspace{-0.1cm}
%	\begin{equation}
%	\label{proof_lemma_divPRRFES_devreg_uppbound_fullproof_eq_0}
	$\min\{L(m,m), L(\bm,m)\}\le\log_2\log_2(4/(\bv - v^m))$.
%	$\min\{l^m_{I^m}, l^{\bm}_{I^m}\}\le\log_2\log_2(4/(\bv - v^m))$.
%	\vspace{-0.1cm}
%	\end{equation}
	To bound the subhorizon $I^m$, decompose it into the numbers of  exploration,  reject-penalization, and exploitation rounds in each phase $l=0,\ldots, L(m',m)$ played by a buyer $m'\!\in\!\{m,\bm\}$. 
	Applying techniques similar to the ones used in the proof of Lemma~\ref{lemma_divPRRFES_indreg_uppbound} (in particular, the bound on the number of exploration rounds: $K^{m'}_l \!\le \!2 \cdot 2^{2^{l-1}}$), we get:
%	\vspace{-0.2cm}
%	\begin{equation}
%	\label{proof_lemma_divPRRFES_devreg_uppbound_fullproof_eq_2}
	$I^m  \!\le\! (L(m',m)\!+\!1) r \!+\!  2\cdot 3 \cdot 2^{2^{L(m',m)}}$
%	\vspace{-0.1cm}
%	\end{equation}
for $m'\!\!\in\!\!\{m,\bm\}$.
This combined with the previous inequality implies Eq.~(\ref{lemma_divPRRFES_devreg_uppbound_eq1}).
	The full proof is in  App.~A.2.2.
\end{proof}
\vspace{-0.2cm}
%Eq.~(\ref{lemma_divPRRFES_devreg_uppbound_eq1}) is quite intuitive: the closer the buyer $m$'s valuation $v^m$ to the maximal one, the longer he is suspected in periods.
This lemma implies the upper bound for  the deviation part of the strategic regret  of the divPRRFES pricing algorithm $\A=\diva_M(\langle\A_1\rangle,\sr_{\A_1})$ against the strategic buyers $\Bset$:
%\vspace{-0.1cm}
$
\SReg^{\dev}(T,\A,\vb,\gb,\betab) \!=\!\! \sum_{m=1}^M I^m (\bv - v^m) \!\le\!  (24+5r)(M-1).
%\vspace{-0.1cm}
$
Let us denote by $\bBset \!:=\! \{m\!\in\!\Bset \mid v^m \!= \!\bv\}$
the set of bidders with the maximal valuation and by $\bbv:=\max_{m\in\Bset\setminus \bBset}v^m$ the highest valuation among non-maximal ones.
Thus, we showed that learning of the max-valuation bidders $\bBset$ converges with the rate inversely proportional to $\bv-\bbv$ (i.e., after the period $\lceil(24+5r)/(\bv-\bbv)\rceil$ the set of suspected bidders is always $\Sset_i = \bBset$) and this learning contributes a constant (w.r.t.\ the horizon $T$) to the strategic regret.
%----
%
%{\bf say smwhr:} See that the individual part of SReg is the rate of learning individual valuations by means of the subalgorithm (PRRFES in our case), while the deviation part of SReg is the rate of learning which bidders have the maximal valuations $\bv${\bf TODO: maybe move it to the section of div class/transformation}.
Finally, Lemma~1, 2, and~3 trivially imply (see App.~A.2.3)  the following theorem.
\begin{theorem}
	\label{th_divPRRFES_sreg_uppbound}
	Let $\g_0\in(0,1)$, $\A_1$ be the PRRFES  algorithm with $r \ge r_{\g_0}$ 
	%TODO: maybe return (link to Eq on r_\g):
%	from Eq.~(\ref{eq_penal_round_lowbound_on_gamma})
	% := \big\lceil \log_{\g_0}\big((1-\g_0)/2\big) \big\rceil$
	and the exploitation rate $g(l)\! =\! 2^{2^{l}}, l\!\in\!\mathbb{Z}_+$, and $\sr_{\A_1}$ be the stopping rule   defined in Eq.(\ref{eq_def_sr_of_PRRFES}).
	Then, for $T\!\!\ge \!\!2$, the strategic regret of the divPRRFES pricing algorithm $\A\!=\!\diva_M(\langle\A_1\rangle\!,\sr_{\A_1}\!)$ against the buyers $\Bset$ is upper bounded: 
%	\vspace{-0.1cm}
	\begin{equation}
%	\begin{split}
	\label{th_divPRRFES_sreg_uppbound_eq1} 
%	&
	\SReg(T,\A,\vb,\!\gb,\!\betab)  \!\le\!  M(r\bv\!+\!4)(\log_2\log_2 T \!+\! 2) +
%	\\&+
	 (24\!+\!5r)(M\!-\!1) \:\:\:\:\:\:\:\: \fa\gb\!\!\in\!\!(0,\!\g_0]^M  \fa\vb\!\!\in\!\! [0,\!1]^M \fa \betab.
%	\end{split}
%	\vspace{-0.1cm}
	\end{equation}
\end{theorem}

\section{Discussion, extensions of the result, and conclusions}
\label{sec_Disc}

{\bf Other auction formats.}
The techniques and algorithms developed in our work can be applied in repeated auctions where another format of selling a good in a round is used.
Namely, our results hold  in our repeated setting with an auction format (within rounds) that satisfies the following: (a) personal reserve prices are allowed; and (b) if a buyer $m$ is only one non-eliminated participant in a round $t$, then his bidding mechanism allows him to choose between getting the good for the reserve price $p^m_t$ and rejecting it. 
This holds e.g.\ for first(/third/..)-price auctions, for PPA with multiple bidders, etc.

\vspace{-0.1cm}
{\bf Regret dependence on $M$.}
The upper bound of the divPRRFES regret in Eq.~(\ref{th_divPRRFES_sreg_uppbound_eq1}) linearly depends on $M$. We believe  that it is not an artifact of our analysis tools, but a payment for the $\diva$-transformation.
Consider the case in which all bidders have the same valuation, i.e.,  all their valuations are $\bv$. Each bidder will be always suspected by divPRRFES (i.e., be in $\Sset_i \fa i$). Hence, divPRRFES will just learn the valuation $\bv$ for each of $M$ bidders independently and, thus, $M$ times slower; i.e., it is natural that the regret of divPRRFES is $M$ times larger than the regret of PRRFES against a single buyer.
However, there might exist an algorithm that do not suffer from dividing structure  in this way.
So, existence of an algorithm with a more favorable regret dependence on $M$ is an open research question.

\vspace{-0.1cm}
{\bf Lower bound and optimality.}
For the case $M=1$, there does exist the lower bound: the strategic regret of any pricing algorithm is  $\Omega(\log\log T)$~\cite{2014-NIPS-Mohri}. Hence, our upper bound for the algorithm divPRRFES is optimal in the general case of any number of bidders.
Nonetheless, structure of the game with non-single buyer ($M\ge2$) is much more complicated, since a buyer has to act in the presence of rivals and under uncertainty about the future. This is an additional opportunity that can be exploited by a pricing algorithm. Thus, the validity of the lower bound $\Omega(\log\log T)$ for $M\ge 2$  is an open research question.
%7. Построили сразу horizon-independent алгоритм
%{\bf Horizon independence.} The algorithm divPRRFES is horizon-independent since it is based on the horizon-independent PRRFES $\A_1$, which induces the subalgorithm $\langle\A_1\rangle$ and the stopping rule $\sr_{\A_1}$. Hence, the seller is not required to know in advance the number of rounds~$T$ of the game, when she applies divPRRFES.
Several other discussions of the results are also in App.~D.

\section{Conclusions}
\label{sec_Conclusions}
%{\bf \underline{Conclusions.}}
We studied the scenario of repeated second-price auctions with reserve pricing where a seller interacts with multiple strategic buyers. Each buyer participates in each round of the game, holds a fixed private valuation for a good, and seeks to maximize his expected future discounted surplus.
%; while the seller seeks a no-regret online learning algorithm to set reserve prices for worst-case valuations.
First, we proposed the so-called dividing transformation that upgrades an algorithm designed for the setup with a single buyer to the multi-buyer case.
Second, the transformation allowed us to obtain a novel horizon-independent algorithm that can be applied against strategic buyers  with regret upper bound of $O(\log\log T)$.
Finally, we introduced non-trivial techniques such as 
(a)~the method to locate the valuation of a strategic buyer in a played round under buyer uncertainty about the future; 
(b)~the decomposition of strategic regret into the individual and deviation parts; and
(c)~the approach to learn the highest-valuation bidder with deviation regret of $O(1)$.

\newpage

\nocite{2009-Book-Krishna}
\nocite{2009-EC-Aggarwal}
\nocite{2011-WWW-Dutting}
\nocite{2013-HBS-Ashlagi}
\nocite{2014-KDD-Agarwal}
\nocite{2009-AER-Varian}
\nocite{2011-WWW-Celis}
\nocite{2011-EC-Ostrovsky}
\nocite{2013-EC-Thompson}
\nocite{1981-MOR-Myerson}
\nocite{2014-WWW-Gomes}
\nocite{2015-NIPS-Zoghi}

\nocite{2007-IJIO-Varian}
\nocite{2009-WWW-Aggarwal}
\nocite{2009-SIGIR-Zhu}
\nocite{2016-WSDM-Charles}
\nocite{2007-Book-Nisan}
\nocite{2013-SODA-Cesa-Bianchi}
\nocite{2013-IJCAI-He}
\nocite{2013-NIPS-Amin}
\nocite{2014-WWW-Hummel}
\nocite{2014-ICML-Mohri}
\nocite{2014-KDD-Yuan}
\nocite{2014-ECRA-Sun}
\nocite{2014-NIPS-Mohri}
\nocite{2015-UAI-Mohri}
\nocite{2016-WWW-Paes}
\nocite{2016-WWW-Rudolph}
\nocite{2016-EC-Roughgarden}
\nocite{2017-WWW-Drutsa}

\nocite{2011-ICAAMS-Chhabra}
\nocite{2015-TEC-Babaioff}
\nocite{2014-SSRN-Kanoria}

\nocite{1993-JET-Schmidt}
\nocite{2015-SODA-Devanur}
\nocite{2017-EC-Immorlica}
\nocite{fudenberg2006behavior}
\nocite{2018-WWW-Lahaie}
\nocite{2018-ArXiV-Vanunts}
\nocite{2018-ICML-Drutsa}

\nocite{2012-RIO-Carare}
\nocite{mirrokni2018optimal}
\nocite{2018-IJGT-Hummel}
\nocite{caillaud2004equilibrium}
\nocite{peters2006internet}
\nocite{2011-ECOMexch-Iyer}
\nocite{list1999price}
\nocite{golrezaei2017boosted}
\nocite{agrawal2018robust}
\nocite{cesa2018dynamic}
\nocite{mirrokni2017non}
\nocite{2017-NIPS-Baltaoglu}
\nocite{2018-NIPS-Huang}
\nocite{bikhchandani1988reputation}

\nocite{hartline2009simple}
\nocite{fu2012ad}
\nocite{2015-NIPS-Morgenstern}

\nocite{hart1988contract}
\nocite{fudenberg2006behavior}
\nocite{2018-NIPS-Mao}
\nocite{2018-NIPS-Leme}

\subsubsection*{Acknowledgments}
	I would like to thank  Sergei Izmalkov who inspired me to conduct this study.
%Use unnumbered third level headings for the acknowledgments. All acknowledgments
%go at the end of the paper. Do not include acknowledgments in the anonymized
%submission, only in the final paper.

%	\newpage
%\section*{References}

\small

\bibliographystyle{abbrv}

\bibliography{2019-arxiv-rspawr}

\begin{thebibliography}{10}

\bibitem{2014-KDD-Agarwal}
D.~Agarwal, S.~Ghosh, K.~Wei, and S.~You.
\newblock Budget pacing for targeted online advertisements at linkedin.
\newblock In {\em KDD'2014}, pages 1613--1619, 2014.

\bibitem{2009-EC-Aggarwal}
G.~Aggarwal, G.~Goel, and A.~Mehta.
\newblock Efficiency of (revenue-) optimal mechanisms.
\newblock In {\em EC'2009}, pages 235--242, 2009.

\bibitem{2009-WWW-Aggarwal}
G.~Aggarwal, S.~Muthukrishnan, D.~P{\'a}l, and M.~P{\'a}l.
\newblock General auction mechanism for search advertising.
\newblock In {\em WWW'2009}, pages 241--250, 2009.

\bibitem{agrawal2018robust}
S.~Agrawal, C.~Daskalakis, V.~Mirrokni, and B.~Sivan.
\newblock Robust repeated auctions under heterogeneous buyer behavior.
\newblock {\em arXiv preprint arXiv:1803.00494}, 2018.

\bibitem{2011-COLT-Amin}
K.~Amin, M.~Kearns, and U.~Syed.
\newblock Bandits, query learning, and the haystack dimension.
\newblock In {\em COLT}, pages 87--106, 2011.

\bibitem{2013-NIPS-Amin}
K.~Amin, A.~Rostamizadeh, and U.~Syed.
\newblock Learning prices for repeated auctions with strategic buyers.
\newblock In {\em NIPS'2013}, pages 1169--1177, 2013.

\bibitem{2014-NIPS-Amin}
K.~Amin, A.~Rostamizadeh, and U.~Syed.
\newblock Repeated contextual auctions with strategic buyers.
\newblock In {\em NIPS'2014}, pages 622--630, 2014.

\bibitem{2016-EC-Ashlagi}
I.~Ashlagi, C.~Daskalakis, and N.~Haghpanah.
\newblock Sequential mechanisms with ex-post participation guarantees.
\newblock In {\em EC'2016}, 2016.

\bibitem{2013-HBS-Ashlagi}
I.~Ashlagi, B.~G. Edelman, and H.~S. Lee.
\newblock Competing ad auctions.
\newblock {\em Harvard Business School NOM Unit Working Paper}, (10-055), 2013.

\bibitem{2015-TEC-Babaioff}
M.~Babaioff, S.~Dughmi, R.~Kleinberg, and A.~Slivkins.
\newblock Dynamic pricing with limited supply.
\newblock {\em ACM Transactions on Economics and Computation}, 3(1):4, 2015.

\bibitem{2016-EC-Balseiro}
S.~Balseiro, O.~Besbes, and G.~Y. Weintraub.
\newblock Dynamic mechanism design with budget constrained buyers under limited
  commitment.
\newblock In {\em EC'2016}, 2016.

\bibitem{2015-ManagSci-Balseiro}
S.~R. Balseiro, O.~Besbes, and G.~Y. Weintraub.
\newblock Repeated auctions with budgets in ad exchanges: Approximations and
  design.
\newblock {\em Management Science}, 61(4):864--884, 2015.

\bibitem{2017-NIPS-Baltaoglu}
M.~S. Baltaoglu, L.~Tong, and Q.~Zhao.
\newblock Online learning of optimal bidding strategy in repeated
  multi-commodity auctions.
\newblock In {\em Advances in Neural Information Processing Systems}, pages
  4507--4517, 2017.

\bibitem{bikhchandani1988reputation}
S.~Bikhchandani.
\newblock Reputation in repeated second-price auctions.
\newblock {\em Journal of Economic Theory}, 46(1):97--119, 1988.

\bibitem{caillaud2004equilibrium}
B.~Caillaud and C.~Mezzetti.
\newblock Equilibrium reserve prices in sequential ascending auctions.
\newblock {\em Journal of Economic Theory}, 117(1):78--95, 2004.

\bibitem{2012-RIO-Carare}
O.~Carare.
\newblock Reserve prices in repeated auctions.
\newblock {\em Review of Industrial Organization}, 40(3):225--247, 2012.

\bibitem{2011-WWW-Celis}
L.~E. Celis, G.~Lewis, M.~M. Mobius, and H.~Nazerzadeh.
\newblock Buy-it-now or take-a-chance: a simple sequential screening mechanism.
\newblock In {\em WWW'2011}, pages 147--156, 2011.

\bibitem{cesa2018dynamic}
N.~Cesa-Bianchi, T.~Cesari, and V.~Perchet.
\newblock Dynamic pricing with finitely many unknown valuations.
\newblock {\em arXiv preprint arXiv:1807.03288}, 2018.

\bibitem{2013-SODA-Cesa-Bianchi}
N.~Cesa-Bianchi, C.~Gentile, and Y.~Mansour.
\newblock Regret minimization for reserve prices in second-price auctions.
\newblock In {\em SODA'2013}, pages 1190--1204, 2013.

\bibitem{2016-WSDM-Charles}
D.~Charles, N.~R. Devanur, and B.~Sivan.
\newblock Multi-score position auctions.
\newblock In {\em WSDM'2016}, pages 417--425, 2016.

\bibitem{2016-SODA-Chawla}
S.~Chawla, N.~R. Devanur, A.~R. Karlin, and B.~Sivan.
\newblock Simple pricing schemes for consumers with evolving values.
\newblock In {\em SODA'2016}, pages 1476--1490, 2016.

\bibitem{2015-EC-Chen}
Y.~Chen and V.~F. Farias.
\newblock Robust dynamic pricing with strategic customers.
\newblock In {\em EC'2015}, pages 777--777, 2015.

\bibitem{2011-ICAAMS-Chhabra}
M.~Chhabra and S.~Das.
\newblock Learning the demand curve in posted-price digital goods auctions.
\newblock In {\em ICAAMS'2011}, pages 63--70, 2011.

\bibitem{2016-EC-Cohen}
M.~C. Cohen, I.~Lobel, and R.~Paes~Leme.
\newblock Feature-based dynamic pricing.
\newblock In {\em EC'2016}, 2016.

\bibitem{2015-SORMS-den-Boer}
A.~V. den Boer.
\newblock Dynamic pricing and learning: historical origins, current research,
  and new directions.
\newblock {\em Surveys in operations research and management science},
  20(1):1--18, 2015.

\bibitem{2015-SODA-Devanur}
N.~R. Devanur, Y.~Peres, and B.~Sivan.
\newblock Perfect bayesian equilibria in repeated sales.
\newblock In {\em SODA'2015}, pages 983--1002, 2015.

\bibitem{2017-WWW-Drutsa}
A.~Drutsa.
\newblock Horizon-independent optimal pricing in repeated auctions with
  truthful and strategic buyers.
\newblock In {\em WWW'2017}, pages 33--42, 2017.

\bibitem{2017-ArXiV-Drutsa}
A.~Drutsa.
\newblock On consistency of optimal pricing algorithms in repeated posted-price
  auctions with strategic buyer.
\newblock {\em CoRR}, abs/1707.05101, 2017.

\bibitem{2018-ICML-Drutsa}
A.~Drutsa.
\newblock Weakly consistent optimal pricing algorithms in repeated posted-price
  auctions with strategic buyer.
\newblock In {\em ICML'2018}, pages 1318--1327, 2018.

\bibitem{2011-WWW-Dutting}
P.~D{\"u}tting, M.~Henzinger, and I.~Weber.
\newblock An expressive mechanism for auctions on the web.
\newblock In {\em WWW'2011}, pages 127--136, 2011.

\bibitem{2016-NIPS-Feldman}
M.~Feldman, T.~Koren, R.~Livni, Y.~Mansour, and A.~Zohar.
\newblock Online pricing with strategic and patient buyers.
\newblock In {\em NIPS'2016}, pages 3864--3872, 2016.

\bibitem{fu2012ad}
H.~Fu, P.~Jordan, M.~Mahdian, U.~Nadav, I.~Talgam-Cohen, and S.~Vassilvitskii.
\newblock Ad auctions with data.
\newblock In {\em Algorithmic Game Theory}, pages 168--179. Springer, 2012.

\bibitem{fudenberg2006behavior}
D.~Fudenberg and J.~M. Villas-Boas.
\newblock Behavior-based price discrimination and customer recognition.
\newblock {\em Handbook on economics and information systems}, 1:377--436,
  2006.

\bibitem{golrezaei2017boosted}
N.~Golrezaei, M.~Lin, V.~Mirrokni, and H.~Nazerzadeh.
\newblock Boosted second-price auctions for heterogeneous bidders.
\newblock 2017.

\bibitem{2014-WWW-Gomes}
R.~Gomes and V.~Mirrokni.
\newblock Optimal revenue-sharing double auctions with applications to ad
  exchanges.
\newblock In {\em WWW'2014}, pages 19--28, 2014.

\bibitem{hart1988contract}
O.~D. Hart and J.~Tirole.
\newblock Contract renegotiation and coasian dynamics.
\newblock {\em The Review of Economic Studies}, 55(4):509--540, 1988.

\bibitem{hartline2009simple}
J.~D. Hartline and T.~Roughgarden.
\newblock Simple versus optimal mechanisms.
\newblock In {\em Proceedings of the 10th ACM conference on Electronic
  commerce}, pages 225--234. ACM, 2009.

\bibitem{2013-IJCAI-He}
D.~He, W.~Chen, L.~Wang, and T.-Y. Liu.
\newblock A game-theoretic machine learning approach for revenue maximization
  in sponsored search.
\newblock In {\em IJCAI'2013}, pages 206--212, 2013.

\bibitem{2016-ICML-Heidari}
H.~Heidari, M.~Mahdian, U.~Syed, S.~Vassilvitskii, and S.~Yazdanbod.
\newblock Pricing a low-regret seller.
\newblock In {\em ICML'2016}, pages 2559--2567, 2016.

\bibitem{2018-NIPS-Huang}
Z.~Huang, J.~Liu, and X.~Wang.
\newblock Learning optimal reserve price against non-myopic bidders.
\newblock In {\em Advances in Neural Information Processing Systems}, pages
  2042--2052, 2018.

\bibitem{2018-IJGT-Hummel}
P.~Hummel.
\newblock Reserve prices in repeated auctions.
\newblock {\em International Journal of Game Theory}, 47(1):273--299, 2018.

\bibitem{2014-WWW-Hummel}
P.~Hummel and P.~McAfee.
\newblock Machine learning in an auction environment.
\newblock In {\em WWW'2014}, pages 7--18, 2014.

\bibitem{2017-EC-Immorlica}
N.~Immorlica, B.~Lucier, E.~Pountourakis, and S.~Taggart.
\newblock Repeated sales with multiple strategic buyers.
\newblock In {\em EC'2017}, pages 167--168, 2017.

\bibitem{2011-ECOMexch-Iyer}
K.~Iyer, R.~Johari, and M.~Sundararajan.
\newblock Mean field equilibria of dynamic auctions with learning.
\newblock {\em ACM SIGecom Exchanges}, 10(3):10--14, 2011.

\bibitem{2014-SSRN-Kanoria}
Y.~Kanoria and H.~Nazerzadeh.
\newblock Dynamic reserve prices for repeated auctions: Learning from bids.
\newblock 2014.

\bibitem{2003-FOCS-Kleinberg}
R.~Kleinberg and T.~Leighton.
\newblock The value of knowing a demand curve: Bounds on regret for online
  posted-price auctions.
\newblock In {\em Foundations of Computer Science}, pages 594--605, 2003.

\bibitem{2009-Book-Krishna}
V.~Krishna.
\newblock {\em Auction theory}.
\newblock Academic press, 2009.

\bibitem{2018-WWW-Lahaie}
S.~Lahaie, A.~M. Medina, B.~Sivan, and S.~Vassilvitskii.
\newblock Testing incentive compatibility in display ad auctions.
\newblock In {\em WWW'2018}, 2018.

\bibitem{2018-NIPS-Leme}
R.~P. Leme and J.~Schneider.
\newblock Contextual search via intrinsic volumes.
\newblock {\em arXiv preprint arXiv:1804.03195}, 2018.

\bibitem{leme2012sequential}
R.~P. Leme, V.~Syrgkanis, and {\'E}.~Tardos.
\newblock Sequential auctions and externalities.
\newblock In {\em SODA'2012}, pages 869--886. SIAM, 2012.

\bibitem{2015-NIPS-Lin}
T.~Lin, J.~Li, and W.~Chen.
\newblock Stochastic online greedy learning with semi-bandit feedbacks.
\newblock In {\em NIPS'2015}, pages 352--360, 2015.

\bibitem{list1999price}
J.~A. List and J.~F. Shogren.
\newblock Price information and bidding behavior in repeated second-price
  auctions.
\newblock {\em American Journal of Agricultural Economics}, 81(4):942--949,
  1999.

\bibitem{2018-NIPS-Mao}
J.~Mao, R.~Leme, and J.~Schneider.
\newblock Contextual pricing for lipschitz buyers.
\newblock In {\em Advances in Neural Information Processing Systems}, pages
  5648--5656, 2018.

\bibitem{2017-NIPS-Medina}
A.~M. Medina and S.~Vassilvitskii.
\newblock Revenue optimization with approximate bid predictions.
\newblock In {\em NIPS'2017}, pages 1856--1864, 2017.

\bibitem{mirrokni2017non}
V.~Mirrokni, R.~Paes~Leme, P.~Tang, and S.~Zuo.
\newblock Non-clairvoyant dynamic mechanism design.
\newblock 2017.

\bibitem{mirrokni2018optimal}
V.~Mirrokni, R.~Paes~Leme, P.~Tang, and S.~Zuo.
\newblock Optimal dynamic auctions are virtual welfare maximizers.
\newblock {\em Available at SSRN}, 2018.

\bibitem{2014-ICML-Mohri}
M.~Mohri and A.~M. Medina.
\newblock Learning theory and algorithms for revenue optimization in second
  price auctions with reserve.
\newblock In {\em ICML'2014}, pages 262--270, 2014.

\bibitem{2015-UAI-Mohri}
M.~Mohri and A.~M. Medina.
\newblock Non-parametric revenue optimization for generalized second price
  auctions.
\newblock In {\em UAI'2015}, 2015.

\bibitem{2014-NIPS-Mohri}
M.~Mohri and A.~Munoz.
\newblock Optimal regret minimization in posted-price auctions with strategic
  buyers.
\newblock In {\em NIPS'2014}, pages 1871--1879, 2014.

\bibitem{2015-NIPS-Morgenstern}
J.~H. Morgenstern and T.~Roughgarden.
\newblock On the pseudo-dimension of nearly optimal auctions.
\newblock In {\em NIPS'2015}, pages 136--144, 2015.

\bibitem{1981-MOR-Myerson}
R.~B. Myerson.
\newblock Optimal auction design.
\newblock {\em Mathematics of operations research}, 6(1):58--73, 1981.

\bibitem{2007-Book-Nisan}
N.~Nisan, T.~Roughgarden, E.~Tardos, and V.~V. Vazirani.
\newblock {\em Algorithmic game theory}, volume~1.
\newblock v.1 CUPC, 2007.

\bibitem{2011-EC-Ostrovsky}
M.~Ostrovsky and M.~Schwarz.
\newblock Reserve prices in internet advertising auctions: A field experiment.
\newblock In {\em EC'2011}, pages 59--60, 2011.

\bibitem{2016-WWW-Paes}
R.~Paes~Leme, M.~P{\'a}l, and S.~Vassilvitskii.
\newblock A field guide to personalized reserve prices.
\newblock In {\em WWW'2016}, 2016.

\bibitem{peters2006internet}
M.~Peters and S.~Severinov.
\newblock Internet auctions with many traders.
\newblock {\em Journal of Economic Theory}, 130(1):220--245, 2006.

\bibitem{2016-EC-Roughgarden}
T.~Roughgarden and J.~R. Wang.
\newblock Minimizing regret with multiple reserves.
\newblock In {\em EC'2016}, pages 601--616, 2016.

\bibitem{2016-WWW-Rudolph}
M.~R. Rudolph, J.~G. Ellis, and D.~M. Blei.
\newblock Objective variables for probabilistic revenue maximization in
  second-price auctions with reserve.
\newblock In {\em WWW'2016}, 2016.

\bibitem{1993-JET-Schmidt}
K.~M. Schmidt.
\newblock Commitment through incomplete information in a simple repeated
  bargaining game.
\newblock {\em Journal of Economic Theory}, 60(1):114--139, 1993.

\bibitem{2014-ECRA-Sun}
Y.~Sun, Y.~Zhou, and X.~Deng.
\newblock Optimal reserve prices in weighted gsp auctions.
\newblock {\em Electronic Commerce Research and Applications}, 13(3):178--187,
  2014.

\bibitem{2013-EC-Thompson}
D.~R. Thompson and K.~Leyton-Brown.
\newblock Revenue optimization in the generalized second-price auction.
\newblock In {\em EC'2013}, pages 837--852, 2013.

\bibitem{2018-ArXiV-Vanunts}
A.~Vanunts and A.~Drutsa.
\newblock Optimal pricing in repeated posted-price auctions.
\newblock {\em arXiv preprint arXiv:1805.02574}, 2018.

\bibitem{2007-IJIO-Varian}
H.~R. Varian.
\newblock Position auctions.
\newblock {\em international Journal of industrial Organization},
  25(6):1163--1178, 2007.

\bibitem{2009-AER-Varian}
H.~R. Varian.
\newblock Online ad auctions.
\newblock {\em The American Economic Review}, 99(2):430--434, 2009.

\bibitem{2014-AER-Varian}
H.~R. Varian and C.~Harris.
\newblock The vcg auction in theory and practice.
\newblock {\em The A.E.R.}, 104(5):442--445, 2014.

\bibitem{2016-JMLR-Weed}
J.~Weed, V.~Perchet, and P.~Rigollet.
\newblock Online learning in repeated auctions.
\newblock {\em JMLR}, 49:1--31, 2016.

\bibitem{2014-KDD-Yuan}
S.~Yuan, J.~Wang, B.~Chen, P.~Mason, and S.~Seljan.
\newblock An empirical study of reserve price optimisation in real-time
  bidding.
\newblock In {\em KDD'2014}, pages 1897--1906, 2014.

\bibitem{2009-SIGIR-Zhu}
Y.~Zhu, G.~Wang, J.~Yang, D.~Wang, J.~Yan, J.~Hu, and Z.~Chen.
\newblock Optimizing search engine revenue in sponsored search.
\newblock In {\em SIGIR'2009}, pages 588--595, 2009.

\bibitem{2015-NIPS-Zoghi}
M.~Zoghi, Z.~S. Karnin, S.~Whiteson, and M.~De~Rijke.
\newblock Copeland dueling bandits.
\newblock In {\em NIPS'2015}, pages 307--315.

\end{thebibliography}

\appendix
\numberwithin{definition}{section}
\numberwithin{equation}{section}
\numberwithin{algorithm}{section}
\numberwithin{figure}{section}
\numberwithin{table}{section}

\section{Missed proofs}

\subsection{Missed proofs from Section~3}

\subsubsection{Proof of Lemma~1}
\label{subsubapp_proof_lemma_divalg_reg_decompose}

\begin{proof}
	Let  $\II^m = \{t^m_i\}_{i=1}^{I^m}$ be the set of rounds in which the bidder $m$ is not eliminated by a barrage reserve pricing. Therefore, we have decomposition of the sequence of all rounds into the union of these sets: $\{1,\ldots,T\} = \cup_{m\in\Bset}\II^m$. Note that we also have a splitting in periods $\{1,\ldots,T\} = \cup_{i=1}^{I}\TT_i$ and the intersection $\II^m\cap\TT_i = \{t^m_i\}$ for $m\in\Bset$, $i=1,\ldots,I^m$.

	So, formally, we have
	\begin{equation}
	\label{proof_lemma_divalg_reg_decompose_fullproof_eq_0}
	\SReg(T,\A,\vb,\gb,\betab) = \Reg\big(T,\A,\vb,\obb_{1:T}(T,\A,\vb,\gb,\betab)\big) \!=\! \sum_{t=1}^T(\bv - \ba_t\bp_t) \!= \!
	\sum_{m\in\Bset}\sum_{i=1}^{I^m}(\bv - \ba_{t^m_i}\bp_{t^m_i}),
	\end{equation}
	where the two first identities follow from definitions, while the latter one is just a change of the order of summation (since $\{1,\ldots,T\} = \cup_{m\in\Bset}\II^m = \cup_{m\in\Bset}\{t^m_i\}_{i=1}^{I^m}$). 
	The terms in the sum could be decomposed in the following way: $\bv - \ba_{t^m_i}\bp_{t^m_i} = \bv -v^m + v^m-  \ba_{t^m_i}\bp_{t^m_i}$.
	Also note, since, in each round $t^m_i$, the bidders $\Bset^{-m}$ are eliminated by a barrage reserve price, then the allocation indicator $\ba_{t^m_i}$ and the transferred payment $\bp_{t^m_i} $   depend only on the behavior of the bidder $m$ in this round, i.e., $\ba_{t^m_i} =a^m_{t^m_i}$, $\bp_{t^m_i} = \bp^m_{t^m_i}$, and, if $a^m_{t^m_i} = 1$,  $ \bp_{t^m_i} = \bp^m_{t^m_i} = p^m_{t^m_i}$.
	So, we can continue Eq.~(\ref{proof_lemma_divalg_reg_decompose_fullproof_eq_0}):
	\begin{equation}
	\label{proof_lemma_divalg_reg_decompose_fullproof_eq_1}
	\begin{split}
	\SReg(T,\A,\vb,\gb,\betab) & =
	\sum_{m\in\Bset}\sum_{i=1}^{I^m}(\bv -v^m + v^m - \ba_{t^m_i}\bp_{t^m_i}) 
	\\&= \sum_{m=1}^M\sum_{i=1}^{I^m}(\bv -v^m) + \sum_{m=1}^M\sum_{i=1}^{I^m}(v^m - a^m_{t^m_i}p^m_{t^m_i}) 
	\\&= \sum_{m=1}^MI^m(\bv -v^m) + \sum_{m=1}^M\Reg^m(\II^m,\A^m,v^m,\ob^m_{1:T}),
	\\&= \SReg^{\dev}(T,\A,\vb,\gb,\betab)+ \SReg^{\ind}(T,\A,\vb,\gb,\betab).
	\end{split}
	\end{equation}

\end{proof}

\subsubsection{Proof of Proposition~1}
\label{subsubapp_proof_prop_reject_bound_divalg}
\begin{proof}
	Let  $t$ be the round in which the bidder $m$ reaches the node $\n$ and rejects his reserve price $p^{m}_t$, which is equal to $p^{m}_t=\p(\n)$ by the construction of the algorithm $\diva_M(\langle\A_1\rangle,\sr)$.
	Note that, in the round $t$, all other bidders $\Bset^{-m}$ are eliminated by a barrage price and the reserve prices set by the $\diva$-algorithm  $\diva_M(\langle\A_1\rangle,\sr)$ depend only on $\ab_{1:T}$ (because $\A_1\in\ASet^{\RPPA}$ and $\sr:\Bset\times\T(\A_1)^M\to \mathtt{bool}$).
	Therefore, it is easy to see that, for any strategy $\sig$, the expected future surplus $\Sur_{t:T}(\A,\g_m,v^m,h^m_t,\beta^m,\sig)$ of the bidder $m$ as a function of the bid $b^m_t=\sig(h^m_t)$ in the round $t$ depends, in fact, only on the binary decision $a^m_t=\Ind_{\{b^m_t\ge p^m_t\}}$:
	more formally, the expected surplus is constant when the bid $b^m_t$ is changed within $\{b^m_t\ge p^m_t\}$ and is constant when the bid $b^m_t$ is changed within $\{b^m_t < p^m_t\}$.
	Moreover, since the buyers are divided (in the whole game) and $\A_1\in\ASet^{\RPPA}$, if two strategies $\sig'$ and $\sig''\in\StrartSet_T$ do not differ in their binary output, i.e., $\Ind_{\{\sig'(h)\ge p^m_t\}}=\Ind_{\{\sig''(h)\ge p^m_t\}} \fa h\in\HH_{1:T}$, then they have the same future discounted surplus.
	Hence, any strategy can be treated as a map to binary decisions $\{0,1\}$ (instead of $\mathbb{R}_+$).
	Let $\hat{\sig}_a$ denote an optimal strategy among all possible strategies in which the binary decision $a^m_t$ in the round $t$ is $a\in\{0,1\}$, i.e., $\Ind_{\{\hat{\sig}_a(h^m_t)\ge p^m_t\}}=a$ and $\hat{\sig}_a$ maximizes
	$$\Expect[\sum_{s=t}^T\!\g^{s-1}_m\ba^m_s(v^m - \bp^m_s)\mid h^m_t, a^m_t = a, \sig, \beta^m].$$

	Given a strategy $\sig\in\StrartSet_T$, let us denote the future expected surplus when following this strategy by $S^m_t(\sig):=\Sur_{t:T}(\A,\g_m,v^m,h^m_t,\beta^m,\sig)$.
	When the optimal strategy $\osig^m$ (used by the buyer) is {\bf pure}, we directly have
	$S^m_t(\hat{\sig}_1) \le S^m_t(\osig^m) = S^m_t(\hat{\sig}_0)$, since the price $p^{m}_t$ is rejected ($a^m_t=0$) by our strategic buyer.
	In the general case, when the buyer's optimal strategy $\osig^m$ is {\bf mixed}, let $\alpha_0$ be the probability of a reject ($a^m_t=0$) and, thus, $1-\alpha_0$ be the probability of an acceptance ($a^m_t=1$) in this strategy. Since the strategy is optimal, its surplus $S^m_t(\osig^m)=\alpha_0 S^m_t(\hat{\sig}_0)+ (1-\alpha_0)S^m_t(\hat{\sig}_1)$ must  be no lower than the surplus $S^m_t(\hat{\sig}_1)$ of the strategy $\hat{\sig}_1$: $$\alpha_0 S^m_t(\hat{\sig}_0)+ (1-\alpha_0)S^m_t(\hat{\sig}_1) \ge S^m_t(\hat{\sig}_1).$$
	Since the price $p^{m}_t$ is rejected,  the probability $\alpha_0>0$ and, thus, $\alpha_0 S^m_t(\hat{\sig}_0) \ge \alpha_0S^m_t(\hat{\sig}_1)$.
	In any way, we obtain: 
	\begin{equation}
	\label{prop_reject_bound_divalg_proof_eq_0}
	S^m_t(\hat{\sig}_1) \le S^m_t(\hat{\sig}_0).
	\end{equation}
	
	Let us bound each side of this inequality:
	\begin{equation}
	\label{prop_reject_bound_divalg_proof_eq_1}
	\begin{split}
	S^m_t(\hat{\sig}_1) &= \Expect[\sum_{s=t}^T\!\g^{s-1}_m\ba^m_s(v^m - \bp^m_s)\mid h^m_t, a^m_t = 1, \hat{\sig}_1, \beta^m] = \\
	&= \g_m^{t-1}(v^m-\p(\n)) + \Expect[\sum_{s=t+1}^T\!\g^{s-1}_m\ba^m_s(v^m - \bp^m_s)\mid h^m_t, a^m_t = 1, \hat{\sig}_1, \beta^m]  \ge \\
	& \ge\g_m^{t-1}(v^m-\p(\n)),
	\end{split}
	\end{equation}
	where, in the second identity, we used the fact that if the bidder accepts the price $\p(\n)$, then he necessarily gets the good since all other bidders $\Bset^{-m}$ are eliminated by a barrage price in this round $t$ ({\bf it is the key point of the proof!}). In the last inequality, we used that the expected surplus in rounds $s\ge t+1$ is at least non-negative, because the subalgorithm $\A_1\in\mathbf{C_R}$ is right consistent and accepting of the offered price $\p(\m)$ in some reached node $\m\in\T(\A_1)$ s.t. $\p(\m)>v^m$ will thus result in reserve prices for him higher than his valuation $v^m$ in all subsequent rounds as well (so, the buyer has no incentive to get a local negative surplus in a round, because it will result in non-positive surplus in all subsequent rounds).
	
	\begin{equation}
	\label{prop_reject_bound_divalg_proof_eq_2}
	\begin{split}
	S^m_t(\hat{\sig}_0) &= \Expect[\sum_{s=t}^T\!\g^{s-1}_m\ba^m_s(v^m - \bp^m_s)\mid h^m_t, a^m_t = 0, \hat{\sig}_0, \beta^m] = \\
	&= \Expect[\sum_{s=t^m_{i+r}}^T\g^{s-1}_m\ba^m_s(v^m - \bp^m_s)\mid h^m_t, a^m_t=0, \hat{\sig}_0, \beta^m] \le \\
	& \le \sum\limits_{s=t+r}^T\g_m^{s-1}(v^m - \p(\n) + \delta_{\n}^{l})< \frac{\g_m^{t+r-1}}{1-\g_m}(v^m - \p(\n) + \delta_{\n}^{l}),
	\end{split}
	\end{equation}	
	where $i$ is the current period of the $\diva$-algorithm  $\diva_M(\langle\A_1\rangle,\sr)$, i.e., the round $t = t^m_i\in\TT_i$ is such that the buyer $m$ is the non-eliminated participant in this round (see Sec.3). In the second identity, we used the fact that if the bidder rejects the price $p^{m}_t$, then the future rounds $\{t^m_{i+j}\}_{j=1}^{r-1}$ (in which the bidder will be non-eliminated) will be reinforced penalization rounds (and the strategic bidder will reject prices in all of them as well). In the first inequality, we just upper bounded surplus by assuming that only this bidder left among the suspected bidders $\Sset_j, j>i,$ and he receives the lowest possible reserve price from the left subtree $\LL(\n)$ of the node $\n$. The latter inequality is just a simple arithmetic upper bound for the sum of discounts $\sum_{s=t+r}^T\g_m^{s-1}$.
	
	We unite these bounds on $S^m_t(\hat{\sig}_0)$ and $S^m_t(\hat{\sig}_1)$ (i.e., Eq.~(\ref{prop_reject_bound_divalg_proof_eq_0}),~(\ref{prop_reject_bound_divalg_proof_eq_1}),~and~(\ref{prop_reject_bound_divalg_proof_eq_2})), divide by $\g_m^{t-1}$, and get
	\begin{equation}
	\label{prop_reject_bound_divalg_proof_eq_3}
	(v^m - \p(\n))\left( 1 - \frac{\g_m^r}{1-\g_m}\right)    <  \frac{\g_m^{r}}{1-\g_m} \delta_{\n}^{l},
	\end{equation}	
	that implies the inequality claimed by the proposition, since $r > \log_{\g_m}(1-\g_m)$. 
\end{proof}

\subsection{Missed proofs from Section~4}

\subsubsection{Proof of Lemma~2}
\label{subsubapp_proof_lemma_divPRRFES_indreg_uppbound}

\begin{proof}
	%	The proof is fairly similar to the o ne of \cite[Th.5]{2017-WWW-Drutsa}.
	The game has been played and $\obb_{1:T}\!=\!\obb_{1:T}(T,\diva_M(\langle\A_1\rangle,\sr),\vb,\gb,\betab)$ are the resulted optimal bids of the strategic buyers $\Bset$.
	So, let $L^m:=l^m_{I^m}$ be the number of phases conducted by the algorithm during the rounds $\II^m=\{t^m_i\}_{i=1}^{I^m}$ against the strategic buyer $m$. 
	Then we decompose the total individual regret over these rounds into the sum of the phases' regrets: $\Reg^m(\II^m,\langle\A_1\rangle,v^m,\ob^m_{1:T}) = \sum_{l=0}^{L^m}R^m_l$. For the regret $R_l$ at each phase except the last one, the following identity holds:
	\begin{equation}
	\label{th_PRRFES_regret_upper_bound_proof_eq1}
	R^m_l =  \sum\limits_{k=1}^{K^m_l} (v^m-p^m_{l,k}) + rv^m + g(l)(v^m-p^m_{l,K^m_{l}}), \quad l=0,\ldots,L^m-1,
	\end{equation} 
	where the first, second, and third terms correspond to the exploration rounds with acceptance, the reject-penalization rounds, and the exploitation rounds\footnote{Note that the prices at the exploitation rounds $p^m_{l, K^m_{l}}$ are equal to either $0$ or an earlier accepted price, and are thus accepted by the strategic buyer (since the buyer's decisions at these rounds do not affect further pricing of the algorithm divPRRFES).}, respectively.
	Since the basis of the subalgorithm PRRFES $\A_1\in\mathbf{C_R}$ is right-consistent~\cite{2017-WWW-Drutsa}, as discussed in the proof of Proposition~1 (see Appendix~\ref{subsubapp_proof_prop_reject_bound_divalg}), the optimal strategy of the bidder $m$ is non-losing~\cite{2017-WWW-Drutsa}:  the buyer has no incentive to get a local negative surplus in a round, because it will result in non-positive surplus in all subsequent rounds.
	
	Hence, since the price $p^m_{l,K^m_{l}}$ is $0$ or has been accepted, we have  $p^m_{l,K^m_{l}} \le v^m$. 
	Second, since the price $p^m_{l,K^m_{l} + 1}$ is rejected, we have $v^m - p^m_{l,K^m_{l} + 1} < (p^m_{l,K^m_{l} + 1} - p^m_{l,K^m_{l}})=\e_l$ (by Proposition~1 since $\zeta_{r,\g_m}\le1$ for $r\ge r_{\g_0}$ and $\g_m\le \g_0$). Hence, the valuation 
	$v^m\in\big[p^m_{l,K^m_{l}}, p^m_{l,K^m_{l}} + 2\e_{l}\big)$ and all accepted prices $p^m_{l+1,k}, \fa k\le K^m_{l+1}$, from the next phase $l+1$ satisfy: 
	$$
	p^m_{l+1,k}\in[q^m_{l+1},v^m)\subseteq \big[p^m_{l,K^m_{l}}, p^m_{l,K^m_{l}} + 2\e_{l}\big) \quad \fa k\le K^m_{l+1},
	$$ 
	because any accepted price has to be lower than the valuation~$v^m$ for the strategic buyer (whose optimal strategy is locally non-losing one, as we stated above). This infers $K^m_{l+1} <2\e_{l}/\e_{l+1} = 2 N_{l+1},$ where $N_l := \e_{l-1}/\e_{l} = \e^{-1}_{l-1} = 2^{2^{l-1}}$. Therefore, for the phases $l=1,\ldots,L^m$, we have:
	\begin{equation*}
	v^m - p^m_{l,K^m_{l}} < 2\e_{l}; \qquad
	v^m-p^m_{l,k} < \e_{l}\big(2N_{l}-k\big)  \fa k\in\mathbb{Z}_{2N_{l}};
	\end{equation*}
	and
	\begin{equation*}
	\begin{split}
	\sum\limits_{k=1}^{K^m_{l}}(v^m-p^m_{l,k}) < \e_{l}\sum\limits_{k=1}^{2N_{l}-1}\big(2N_{l}-k\big) = \e_{l}\frac{2N_{l}-1}{2}\big(2\cdot 2N_l-2N_{l}\big) 
	\le 2N_l\cdot N_l\e_{l} = 2N_l\cdot \e_{l-1} = 2,
	\end{split}
	\end{equation*} 
	where we used the definitions of $N_l$ and $\e_l$.
	For the zeroth phase $l=0$, one has trivial bound 
	$\sum_{k=1}^{K^m_{0}}(v-p^m_{0,k}) \le 1/2$.
	Hence, by definition of the exploitation rate $g(l)$, we have $g(l)=\e_l^{-1}$ and, thus,
	\begin{equation}
	\label{th_PRRFES_regret_upper_bound_proof_eq2} 
	R^m_l \le  2 + rv^m + g(l) \cdot 2\e_{l} \le  rv^m+4, \quad l=0,\ldots,L-1.
	\end{equation} 
	
	Moreover, this inequality holds for the $L^m$-th phase, since it differs from the other ones only in possible absence of some rounds (reject-penalization or exploitation ones). Namely, for the $L^m$-th phase, we have:
	\begin{equation}
	\label{th_PRRFES_regret_upper_bound_proof_eq3} 
	R^m_L =  \sum\limits_{k=1}^{K^m_L} (v^m-p^m_{L^m,k}) + r_{L^m}v^m + g_{L^m}(L^m)(v^m-p^m_{L^m,K^m_{L^m}}),
	\end{equation}
	where $r_{L^m}$ is the actual number of reject-penalization rounds and $g_{L^m}(L^m)$ is the actual number of exploitation ones in the last phase. Since $r_{L^m}\le r$ and $g_{L^m}(L^m)\le g(L^m)$,  
	the right-hand side of Eq.~(\ref{th_PRRFES_regret_upper_bound_proof_eq3}) is upper-bounded by the right-hand side of Eq.~(\ref{th_PRRFES_regret_upper_bound_proof_eq1}) with $l=L^m$, which is in turn upper-bounded by the right-hand side of Eq.~(\ref{th_PRRFES_regret_upper_bound_proof_eq2}).
	%	 but this absence can be easily upper-bounded by the regret of a possible  $L$-th phase as if all these rounds are played in.
	Finally, one has
	\begin{equation*}
	\Reg^m(\II^m,\diva_M(\langle\A_1\rangle,\sr),v^m,\ob^m_{1:T}) = \sum_{l=0}^{L^m}R^m_l \le \left(rv^m+4\right)(L^m+1).
	\end{equation*} 
	Thus, one needs only to estimate the number of phases $L^m$ by the subhorizon $I^m$. So, for $2\le I^m\le 2 + r + g(0)$, we have $L^m=0$ or $1$ and thus $L^m+1\le 2 \le \log_2\log_2 I^m + 2$. For $I^m \ge 2 + r + g(0)$, we have $I^m = \sum_{l=0}^{L^m-1}(K^m_l + r + g(l)) + K^m_{L^m} + r_{L^m} + g_{L^m}(L^m)\ge g(L^m-1)$ with $L^m>0$.
	%	(when $v<1$, otherwise Eq.~(\ref{th_PRRFES_regret_upper_bound_eq1}) holds).
	Hence, $g(L^m-1) = 2^{2^{L^m-1}}\le I^m$, which is equivalent to $L^m \le \log_2\log_2 I^m + 1$. Summarizing, we get the claimed upper bound of the lemma.
\end{proof}

\subsubsection{Proof of Lemma~3}
\label{subsubapp_proof_lemma_divPRRFES_devreg_uppbound}

\begin{proof}
	Let $\bm\in\bBset$ be one of the bidders $\bBset=\{m\!\in\!\Bset \mid v^m \!= \!\bv\}$ that have the maximal valuation $\bv$.
	Then, the stopping rule $\sr_{\A_1}$ (which is based on the rule $\rho(m, \lb, \qb):=
	\exists \hat m\in\Bset^{-m}:  q^m + 2\e_{l^m-1}<q^{\hat m} \fa \lb\in\mathbb{Z}^M_+, 
	\fa \qb\in\mathbb{R}^M_+$) is executed no later than the period $i'$ where the upper bound $q_{l_{i'}^m}^m + 2\e_{l_{i'}^m-1}$ of the bidder $m$'s valuation becomes lower than the lower bound $q_{l_{i'}^{\bm}}^{\bm} $ of the bidder $\bm$'s valuation\footnote{Note that it is correct to consider $l_{i}^m$ in any period $i$ even though the buyer $m$ is not suspected in this period, i.e., $m\notin\Sset_i$. This is because the algorithm stops change the tracking node $\n_i^m$ in the subalgorithm tree $\T(\langle\A_1\rangle)$ after the period $I^m$, but  $l_{i}^m$ just remains the same in all subsequent periods, i.e., we formally set $l_{i}^m=l_{I^m}^m$ for all $i>I^m$.}.
	
	Moreover, since $v^m\in [q_{l_j^m}^m, q_{l_j^m}^m + 2\e_{l_j^m-1}]$ and $v^{\bm}\in [q_{l_j^{\bm}}^{\bm}, q_{l_j^{\bm}}^{\bm} + 2\e_{l_j^{\bm}-1}]$ for any period $j$, the stopping rule is executed no later than the period $i$ where 
	both the phase iteration parameter $\e_{l^m_i}$ of the bidder $m$ and the phase iteration parameter $\e_{l^{\bm}_i}$ of the bidder $\bm$ become smaller than one quarter of the difference between the valuations of these bidders, i.e.,  $\e_{l^m_i}$ and $\e_{l^{\bm}_i}<\frac{\bv - v^m}{4}$ (because, in this case, the segments $[q_{l_i^m}^m, q_{l_i^m}^m + 2\e_{l_i^m-1}]$ and $ [q_{l_i^{\bm}}^{\bm}, q_{l_i^{\bm}}^{\bm} + 2\e_{l_i^{\bm}-1}]$ do not intersect at all, what implies $q_{l_i^m}^m + 2\e_{l_i^m-1}< q_{l_i^{\bm}}^{\bm}$).
	
	Therefore, in the periods $i \le I^m$, it is not possible to have simultaneously $\e_{l^m_{i}}<\frac{\bv - v^m}{4}$ and $\e_{l^{\bm}_{i}}<\frac{\bv - v^m}{4}$. So, in the period $i=I^m$, either $\e_{l^m_{I^m}}\ge\frac{\bv - v^m}{4}$, or (not exclusively) $\e_{l^{\bm}_{I^m}}\ge\frac{\bv - v^m}{4}$ holds.
	In particular, from the definition of the phase iteration parameter $\e_l=2^{-2^{l}}$, we have: if $\e_{l}  \ge \delta$ for some $l\in\mathbb{Z}_+$ and $\delta\in(0,1/2)$, then 
	$$
	\e_{l} = 2^{-2^{l}}  \ge \delta
	\quad\Leftrightarrow\quad -2^{l} \ge \log_2\delta
	\quad\Leftrightarrow\quad 2^{l} \le \log_2\frac{1}{\delta}
	\quad\Leftrightarrow\quad l \le \log_2\log_2\frac{1}{\delta}.
	$$
	Hence, in the period $I^m$, the following holds:
	\begin{equation*}
	l^m_{I^m}\le\log_2\log_2\frac{4}{\bv - v^m} \quad \hbox{or (not exclusively)} \quad l^{\bm}_{I^m}\le\log_2\log_2\frac{4}{\bv - v^m},
	\end{equation*}
	and,  thus,
	\begin{equation}
	\label{proof_lemma_divPRRFES_devreg_uppbound_fullproof_eq_0}
	\min\{l^m_{I^m}, l^{\bm}_{I^m}\}\le\log_2\log_2\frac{4}{\bv - v^m}.
	\end{equation}
	
	Finally, we bound $I^m$. Let, $L^{m';m}:=l^{m'}_{I^m}$ be the phase of a buyer $m' \in \{m, \bm\}$ in the period $I^m$. As in the proof of Lemma~2 (see Appendix~\ref{subsubapp_proof_lemma_divPRRFES_indreg_uppbound}) we decompose $I^m$  into the numbers of exploration,  reject-penalization, and exploitation rounds in each phase $l=0,\ldots, L^{m';m}$ passed by the buyer $m'$. Namely,
	\begin{equation}
	\label{proof_lemma_divPRRFES_devreg_uppbound_fullproof_eq_1}
	I^m = \sum_{l=0}^{L^{m';m} - 1 } (K^{m'}_l + r + g(l)) + K^{m'}_{L^{m';m}} + r^{m'}_{L^{m';m}} + g^{m'}_{L^{m';m}}, 
	\end{equation}
	where $r^{m'}_l$ and $g^{m'}_l$ are the numbers of  penalization rounds and exploitation rounds, resp., passed by the buyer $m'$ in the last phase $l=L^{m';m}$ before reaching the period $I^m$. Let us trivially bound $r^{m'}_{L^{m';m}}\le r$ and $g^{m'}_{L^{m';m}} \le g(L^{m';m})$.
	We also know that, for any $l\in\mathbb{Z}_+$, $K^{m'}_l \le 2 \cdot 2^{2^{l-1}}$ (see the proof of Lemma~2 in Appendix~\ref{subsubapp_proof_lemma_divPRRFES_indreg_uppbound}). Therefore, Eq.~\ref{proof_lemma_divPRRFES_devreg_uppbound_fullproof_eq_1} implies
	\begin{equation}
	\label{proof_lemma_divPRRFES_devreg_uppbound_fullproof_eq_2}
	I^m \le \sum_{l=0}^{L^{m';m}  } (2\cdot 2^{2^{l-1}}+ r + 2^{2^{l}}) \le \sum_{l=0}^{L^{m';m}  } (3\cdot 2^{2^{l}}+ r) \le (L^{m';m}+1) r +  2\cdot 3 \cdot 2^{2^{L^{m';m}}}, 
	\end{equation}
	
	Taking $m' = m$ and $m' = \bm$,  we get the following from Eq.~(\ref{proof_lemma_divPRRFES_devreg_uppbound_fullproof_eq_2}):
	\begin{equation}
	\label{proof_lemma_divPRRFES_devreg_uppbound_fullproof_eq_3}
	I^m \le (\min\{l^{m}_{I^m},l^{\bm}_{I^m}\}+1) r +  6 \cdot 2^{2^{\min\{l^{m}_{I^m},l^{\bm}_{I^m}\}}} 
	\le r(\log_2\log_2\frac{4}{\bv - v^m} + 1) + 6 \cdot \frac{4}{\bv - v^m} , 
	\end{equation}
	where we used the definition of $L^{m';m}:=l^{m'}_{I^m}$ and the upper bound for the phases $l^{m}_{I^m}$ and $l^{\bm}_{I^m}$ in Eq.~(\ref{proof_lemma_divPRRFES_devreg_uppbound_fullproof_eq_0}).
	So, Eq.~(\ref{proof_lemma_divPRRFES_devreg_uppbound_fullproof_eq_3}) implies the claim of the lemma.

\end{proof}

\subsubsection{Proof of Theorem~1}
\label{subsubapp_proof_th_divPRRFES_sreg_uppbound}

\begin{proof}
	From Lemma~1, we have:
	\begin{equation}
	\label{proof_th_divPRRFES_sreg_uppbound_fullproof_eq_1}
	\SReg(T,\A,\vb,\gb,\betab) = \sum_{m=1}^{M} \Reg^m(\II^m,\A^m,v^m,\ob^m_{1:T}) + \sum_{m=1}^{M} I^m (\bv - v^m).
	\end{equation}
	
	From Lemma~2, if $I^m \ge 2$, one can upper bound the first term in right-hand side of Eq.~(\ref{proof_th_divPRRFES_sreg_uppbound_fullproof_eq_1}) since $\A^m=\langle\A_1\rangle$:
	\begin{equation}
	\label{proof_th_divPRRFES_sreg_uppbound_fullproof_eq_2}
	\Reg^m(\II^m,\A^m,v^m,\ob^m_{1:T}) \le (rv^m+4)(\log_2\log_2 I^m + 2) \le (r\bv+4)(\log_2\log_2 T + 2),
	\end{equation}
	where we bounded the subhorizon $I^m$ of each bidder $m\in\Bset$ by the time horizon $T$ (i.e., $I^m \le T$) and the valuation $v^m$ of each bidder $m\in\Bset$ by the maximal valuation (i.e., $v^m\le \bv$). Note that the latter bound of Eq.~(\ref{proof_th_divPRRFES_sreg_uppbound_fullproof_eq_2}) holds for $\Reg^m(\II^m,\A^m,v^m,\ob^m_{1:T})$ in the case of $I^m = 1$ as well (this case has not been provided by Lemma~2).
	
	From Lemma~3, one can upper bound the second term in right-hand side of Eq.~(\ref{proof_th_divPRRFES_sreg_uppbound_fullproof_eq_1}):
	\begin{equation}
	\label{proof_th_divPRRFES_sreg_uppbound_fullproof_eq_3}
	\sum_{m=1}^{M} I^m (\bv - v^m) \le \sum_{\{m\in\Bset \mid v^m\neq \bv\}}\frac{24+5r}{\bv - v^m} (\bv - v^m) \le  (24+5r)(M-1),
	\end{equation}
	where we used that at least one bidder $\bm\in\Bset$ has $v^{\bm}=\bv$ and, hence, $|\{m\in\Bset \mid v^m\neq \bv\}|\le M-1$.
	
	Thus, plugging Eq.~(\ref{proof_th_divPRRFES_sreg_uppbound_fullproof_eq_2}) and Eq.~(\ref{proof_th_divPRRFES_sreg_uppbound_fullproof_eq_3}) into Eq.~(\ref{proof_th_divPRRFES_sreg_uppbound_fullproof_eq_1}), we obtain the claimed bound for the strategic regret of divPRRFES.
	
\end{proof}

%\newpage

\section{The pseudo-codes}

\subsection{The pseudo-code of $\diva$-transformation}

\begin{algorithm}
	\small
	\caption{Pseudo-code of a $\diva$-transformation $\diva_M(\A_1,\sr)$ of a RPPA algorithm $\A_1\in\ASet^{\RPPA}$.}
	\label{alg_div_transform}
	\begin{algorithmic}[1]
		\STATE {{\bfseries Input:} $M\in\mathbb{N}$,  $\A_1\in\ASet^{\RPPA},$ $\sr:\Bset\times\T(\A_1)^M\to \mathtt{bool}$}
		\STATE {{\bfseries Initialize:} $t:=1$, $\Sset :=\Bset$, $\n[\:] := \{\treeroot(\T(\A_1))\}_{m=1}^M$}
		\WHILE{$t\le T$}
		\FORALL{$m\in\Sset$}		    
		\STATE {Set the price $\p(\n[m])$ as reserve to the buyer $m$}
		\STATE {Set the price $p^{\mathrm{bar}}$  as reserve to the buyers from $\Bset^{-m}$}
		\STATE {$\bb[\:] \gets $ get bids from the buyers $\Bset$}
		\IF{$\bb[m]\ge \p(\n[m])$}
		\STATE {Allocate $t$-th good to the buyer $m$ for the price $\p(\n[m])$}
		\STATE {$\n[m] := \rt(\n[m])$}
		\ELSE
		\STATE {$\n[m] := \lf(\n[m])$}
		\ENDIF
		\STATE {$t := t+1$}
		\IF{ $t> T$}
		\STATE {{\bf break} } 
		\ENDIF
		\ENDFOR
		\STATE {$\Sset^{\mathrm{old}}:=\Sset$}
		\FORALL{$m\in\Sset^{\mathrm{old}}$}		    
		\IF{  $\sr(m,\n[\:])$}
		\STATE {$\Sset:=\Sset\setminus\{m\}$}
		\ENDIF
		\ENDFOR
		\ENDWHILE	
	\end{algorithmic}
\end{algorithm}

\newpage

\subsection{The pseudo-code of divPRRFES}
\begin{algorithm}
	\small
	\caption{Pseudo-code of the algorithm divPRRFES.}
	\label{alg_divPRRFES}
	\begin{algorithmic}[1]
		\STATE {{\bfseries Input:} $M\in\mathbb{N}$,  $r\in\mathbb{N}$, and $g:\mathbb{Z}_{+}\rightarrow\mathbb{Z}_{+}$}
		\STATE {{\bfseries Initialize:} $t:=1$, $\Sset :=\Bset$, $q[\:] := \{0\}_{m=1}^M$, $l[\:] := \{0\}_{m=1}^M$, $x[\:] := \{0\}_{m=1}^M$, $\mathrm{state}[\:] := \{\mathtt{"explore"}\}_{m=1}^M$ }
		\WHILE{$t\le T$}
		\FORALL{$m\in\Sset$}		    
		\IF{$\mathrm{state}[m]=\mathtt{"penalize"}$}
		\STATE {$p:= 1$ {\color{gray} // a reinforced penalization round for the buyer $m$ }}
		\STATE {$x[m]:= x[m]-1$}
		\ENDIF
		\IF{$\mathrm{state}[m]=\mathtt{"explore"}$}
		\STATE {$p:= q[m] + 2^{-2^{l[m]}}$ {\color{gray} // an exploration round for the buyer $m$}}
		\ELSE
		\STATE {$p:= q[m]$ {\color{gray} // an exploitation round for the buyer $m$}}
		\STATE {$x[m]:= x[m]-1$}
		\ENDIF
		\STATE {Set the price $p$ as reserve to the buyer $m$}
		\STATE {Set the price $p^{\mathrm{bar}}$  as reserve to the buyers from $\Bset^{-m}$}
		\STATE {$\bb[\:] \gets $ get bids from the buyers $\Bset$}
		\IF{$\bb[m]\ge p$}
		\STATE {Allocate $t$-th good to the buyer $m$ for the price $p$}
		\STATE {$q[m]:=p$}
		\IF{$\mathrm{state}[m]=\mathtt{"penalize"}$}
		\STATE {$x[m]:=-1$  {\color{gray} // a reinforced penalization price is accepted; set $1$ to the buyer $m$ all his rounds }}
		\ENDIF			
		\ELSE
		\IF{$\mathrm{state}[m]=\mathtt{"explore"}$}
		\STATE {$\mathrm{state}[m]:=\mathtt{"penalize"}$}
		\STATE {$x[m]:= r$ {\color{gray} // an exploration price is rejected; move the buyer $m$ to penalization}}
		\ENDIF
		\ENDIF
		\IF{$\mathrm{state}[m]=\mathtt{"penalize"}$ \AND $x[m]=0$}
		\STATE {$\mathrm{state}[m]:=\mathtt{"exploit"}$}
		\STATE {$x[m]:= g(l[m])$ {\color{gray} // penalization rounds are ended; move the buyer $m$ to exploitation}}
		\ENDIF
		\IF{$\mathrm{state}[m]=\mathtt{"exploit"}$ \AND $x[m]=0$}
		\STATE {$\mathrm{state}[m]:=\mathtt{"explore"}$}
		\STATE {$l[m]:= l[m]+1$ {\color{gray} // exploitation rounds are ended; move the buyer $m$ to the next phase }}
		\ENDIF
		\STATE {$t := t+1$}
		\IF{ $t> T$}
		\STATE {{\bf break} } 
		\ENDIF
		\ENDFOR
		\STATE {$\Sset^{\mathrm{old}}:=\Sset$}
		\STATE {$q_{\max}:=\max_{m\in\Bset}(q[m])$}
		\FORALL{$m\in\Sset^{\mathrm{old}}$}		    
		\IF{  $q[m] + 2*2^{-2^{l[m]-1}} < q_{\max}$}
		\STATE {$\Sset:=\Sset\setminus\{m\}$ {\color{gray} // remove the buyer $m$ from suspected ones if the stopping rule is satisfied}}
		\ENDIF
		\ENDFOR
		\ENDWHILE	
	\end{algorithmic}
\end{algorithm}

\newpage

\section{Summary on used notations}

Note that we use several mnemonic notations: 
\begin{itemize}
	\item upper index for a value of a particular buyer (e.g., $v^m$, $a^m_t$, $p^m_t$, etc.); 
	\item boldface for a vector of values for all bidders (e.g., $\vb$, $\ab_t$, $\pb_t$, etc.); 
	\item bar (overline) for terms associated with  best value / winning (e.g., the winner $\bm_t$, the highest valuation $\bv$, etc.); etc. 
\end{itemize}

The full list of used notations is summarized below in the following tables.

\subsection{General notations}

See Tables~\ref{tbl_notations_general_p1}, ~\ref{tbl_notations_general_p2}, and~\ref{tbl_notations_general_p3}.
\begin{table}
	\centering
	\caption{General notations: part I.}
	\label{tbl_notations_general_p1}
	%		\vspace{-0.3cm}
	\begin{tabular}{|c|l|p{10cm}|}
		\hline
		\textbf{Notation} & \textbf{Expression} & \textbf{Description} \\
		\hline
		\hline

		$\Expect[\cdot]$ & & expectation		\\
		
		\hline

		$\Ind_{B}$ & &   the indicator: $\Ind_{B} = 1$, when $B$ holds, and $0$, otherwise.		\\
		
		\hline
		
		$T$ & & the [time] horizon, the number of rounds in the repeated game		\\
		\hline
		$t$ & & a round in the repeated game, $t\in\{1,\ldots,T\}$		\\

		\hline
		$v^m$ & & the valuation of a buyer $m$	\\
		
		\hline
		$\bv$ & $=\max_{m\in\Bset}v^m$ &  the highest valuation among the buyers		\\
		\hline
		$\bbv$ & $=\max_{m\in\Bset\setminus \bBset}v^m$ & 	the maximal valuation among non-highest valuations ot the buyers (if exists)	\\
		
		\hline
		$\bm$ & & a buyer that has the highest valuation $\bv$		\\
		
		\hline
		$\bm_t$ & $=\argmax_{m\in\Bset_t}b^m_t$ & the winning bidder in a round $t$ for a given play of the game (if exists)\\
		
		\hline
		$b^m_t$ & & the bid of a buyer $m$	in a round $t$ for a given play of the game\\
		
		\hline
		$p^m_t$ & & the reserve price set to a buyer $m$	in a round $t$ for a given play of the game\\
		
		\hline
		$a^m_t$ & $=\Ind_{b^m_t \ge p^m_t}$ &  indicator of bidding higher than the  reserve price by a buyer $m$	in a round $t$ for a given play of the game \\
		
		\hline
		$\ba^m_t$ & $=  \Ind_{\{\Bset_t \neq \eset \&  m = \bm_t\}}$& the allocation outcome of a round $t$ for a bidder $m$	 for a given play of the game \\
		
		\hline
		$\ba_t$ & $= \Ind_{\{\Bset_t \neq \eset\}}$ &  the allocation outcome of a round $t$ over all bidders	 for a given play of the game		\\
		
		\hline
		$\bp^m_t$ &  $= \ba^m_t\bp_t$ & the payment outcome of a round $t$ for a bidder $m$	 for a given play of the game		\\
		\hline
		$\bp_t$ & $=\max\{p^{\bm_t}_t, \max_{m\in\Bset^{-\bm_t}_t} b^m_t \}$ & the payment outcome of a round $t$  over all bidders	 for a given play of the game 	\\

		\hline
		$\mathbf{x}$ &  $=\{x^m\}_{m=1}^M$ & the vector of buyer values of some notion $x$ (e.g., valuations $\vb$, bids $\bb_t$, reserve prices $\pb_t$, payments $\bpb_t$, allocations $\bab_t$ and $\ab_t$   etc)	\\
		
		\hline
		$x_{t_1:t_2}$ &$=\{x_t\}_{t=t_1}^{t_2}$ &  the subseries  of some time series $\{x_t\}_{t=1}^T$ (e.g., bids $\bb_{1:T}$, reserve prices $\pb_{1:T}$, payments $\bpb_{1:T}$, allocations $\bab_{1:T}$ and $\ab_{1:T}$   etc)	\\
		
		\hline
		$\ASet_M$ & & the set of pricing algorithms of the seller	against $M$ buyers	\\
		\hline
		$\ASet^{\RPPA}$ & $\subset \ASet_1$  & the subclass of $1$-buyer pricing algorithms for repeating posted-price auctions\\
		\hline
		$\A$ & & a pricing algorithm (generally, from the set $\ASet_M$)		\\
		
		\hline
		$M$ & & the number of buyers in the repeated game			\\
		\hline
		$\Bset$ & $=\{1,\ldots,M\}$ &  the set of buyers (bidders)		\\
		\hline
		$\bBset$ & $=\{m\in\Bset \mid v^m = \bv\}$ &  the set of buyers whose valuation is the highest one $\bv$ 	\\
		
		\hline
		$\Bset^{-m}$ & $=\Bset \setminus \{m\}$ &  the set of buyers (bidders)	without the buyer $m$	\\
		
		\hline
		$\Bset_t$ & $=\{m\in\Bset \mid b^m_t \ge p^m_t\}$ &  the set of actual buyers in a round $t$ (they bid higher than reserve prices) 	\\

		\hline
		
	\end{tabular}
\end{table}

\begin{table}
	\centering
	\caption{General notations: part II.}
	\label{tbl_notations_general_p2}
	%		\vspace{-0.3cm}
	\begin{tabular}{|c|l|p{10cm}|}
		\hline
		\textbf{Notation} & \textbf{Expression} & \textbf{Description} \\
		\hline
		\hline

		$\Reg(\ldots)$ & & regret of a pricing algorithm		\\
		\hline
		$\SReg(\ldots)$ & & strategic regret of a pricing algorithm		\\
		\hline
		$\Sur(\ldots)$ & & expected surplus of a buyer (bidder) 		\\

		\hline
		$\g_m$ & & the discount rate of a buyer $m\in\Bset$		\\
		\hline
		$\gb$ & $=\{\g_m\}_{m=1}^M$ &  the vector of the discount rates of the buyers	\\

		\hline
		$h$ & & a buyer history		\\
		
		\hline
		$h^m_t$ &$=(b^m_{1:t-1}\!,p^m_{1:t}\!,\ba^m_{1:t-1}\!,\bp^m_{1:t-1})$&  the history available to a buyer $m$ in a round $t$ for a given play of the game	\\
		
		\hline
		$\sig$ & $\in\StrartSet_T$& a buyer strategy		\\

		%\hline
		%$\abeta_{h}^m$ & $=(\amu_h^m,\avsig_h^m)$& the beliefs of a buyer $m$ for a history $h$ 		\\
		%\hline
		%$\amu_{h}^m$ & & the beliefs on the valuations of the other bidders of a buyer $m$ for a history $h$	\\
		%\hline
		%$\avsig_{h}^m$ & & the beliefs on the strategies of the other bidders of a buyer $m$ for a history $h$		\\
		%\hline
		%$\abeta^m$ & $=\{\abeta_h^m\}_{h\in\HH_{1:T}}$ & the beliefs of a buyer $m$ for all possible histories 		\\
		\hline
		$\beta^m$ & $\in\StrartSet^{M-1}_T$& the beliefs of a buyer $m$  on the strategies of the other bidders  		\\
		\hline
		$\betab$ & $=\{\beta^m\}_{m=1}^M$ & the beliefs of all buyers
		% for all possible histories 	
		\\
		
		\hline
		$\HH_t$ & & the set of all possible histories in a round $t$		\\
		
		\hline
		$\HH_{t_1:t_2}$ & $= \!\sqcup_{t=t_1}^{t_2}\HH_t$ & the disjoint union of the sets of histories in rounds $t_1,\ldots,t_2$		\\
		
		\hline
		$\StrartSet_T$ & & the set of all possible buyer strategies 		\\

		\hline
		$\osig^m$ & & an optimal strategy of a buyer $m$  in a  round	$t$	\\
		\hline
		$\ob^m_t $ & & the optimal bid of a buyer $m$ in a  round	$t$	for a given play of the game	\\
		\hline
		$\obb_t$ & $=\{ \ob^m_t\}_{m=1}^M$ & the optimal bids of all buyers  in a  round	$t$	for a given play of the game		\\
		\hline
		$\obb_{1:T}$ & & the optimal bids of all buyers  in all  rounds	for a given play of the game		\\

		\hline
		
	\end{tabular}
\end{table}

\begin{table}
	\centering
	\caption{General notations: part III (related to RPPA algorithms).}
	\label{tbl_notations_general_p3}
	%		\vspace{-0.3cm}
	\begin{tabular}{|c|l|p{10cm}|}
		\hline
		\textbf{Notation} & \textbf{Expression} & \textbf{Description} \\
		\hline
		\hline
		$\T(\A_1)$ & & the complete binary tree associated with a RPPA algorithm $\A_1$		\\
		\hline
		$\n$ or $\m$  & & a node in the complete binary tree $\T(\A_1)$ of a RPPA algorithm $\A_1$	  		\\
		\hline
		$\rt(\n)$ & & the right child of a node $\n$		\\
		\hline
		$\lf(\n)$ & & the left child of a node $\n$		\\
		\hline
		$\RR(\n)$ & & the right subtree of a node $\n$ (its root is $\rt(\n)$)		\\
		\hline
		$\LL(\n)$ & & the left  subtree of a node $\n$	(its root is $\lf(\n)$)		\\
		\hline
		$\treeroot(\T)$ & & the root of a tree $\T$		\\
		\hline
		$\p(\n)$ & & the price in a node $\n$ (that is offered to a buyer when an algorithm reaches this node)		\\
		
		\hline
		$\T_1 \cong \T_2$ & & the trees $\T_1 $ and $\T_2$	are price-equivalent	\\
		\hline
		$\delta_{\n}^{l}$ & $= \p(\n) - \inf_{\m\in\LL(\n)}\p(\m)$ & the left increment of a node $\n$ \\
		
		\hline
	\end{tabular}
\end{table}

\newpage

\subsection{Notations related to dividing algorithms}

See Table~\ref{tbl_notations_div_alg}.

\begin{table}
	\centering
	\caption{Notations related to dividing algorithms.}
	\label{tbl_notations_div_alg}
	%		\vspace{-0.3cm}
	\begin{tabular}{|c|l|p{10cm}|}
		\hline
		\textbf{Notation} & \textbf{Expression} & \textbf{Description} \\
		\hline
		\hline
		
		$i$ & & a period of a dividing algorithm ({\bf do not confuse with} (1) a round of the game and (2) a phase of PRRFES algorithm!)		\\
		\hline
		
		$t^m_i$ & & the round in a period $i$ in which the bidder $m$ is not eliminated by a barrage price (i.e., $m$ is non-eliminated participant) of  a dividing algorithm 		for a given play of the game	\\
		\hline
		
		$p^{m, \mathrm{bar}}$ or $p^{\mathrm{bar}}$ & & a barrage reserve price		\\
		\hline
		
		$\Sset_i$ & & the set of bidders suspected by 	a dividing algorithm in a period $i$ 	for a given play of the game	\\
		\hline
		
		$\TT_i$ &  
		%$= \{\ttau_i+1, \ldots, \ttau_i+\tau_i\}$
		& the rounds of a period $i$ 	for a given play of the game	\\
		\hline
		%$\tau_i$ &  $=|\Sset_i| = |\TT_i|$& the length of a period $i$; the number of rounds (suspected bidders) in the period $i$ 	for a given play of the game	\\
		%\hline
		%$\ttau_i$ & & the number of rounds passed before 	a period $i$ 	for a given play of the game	\\
		%\hline
		$\II^m$ & $= \{t^m_i\}_{i=1}^{I^m}$ & the rounds	in which the bidder $m$ is not eliminated by a barrage price (i.e., $m$ is non-eliminated participant) of dividing algorithm 	for a given play of the game		\\
		\hline
		
		$I^m$ &$= |\II^m|$ & the subhorizon of a buyer $m$ (the number of periods in which he is suspected, i.e., $m\in\Sset_i$)		for a given play of the game	\\
		\hline
		
		$\A^m$ & & the subalgorithm of  a dividing algorithm that acts against a buyer $m$		\\
		\hline
		
		$\Reg^m(\ldots)$ & & Regret of the subalgorithm  of a dividing algorithm that acts against a buyer $m$			\\

		\hline
		$\diva_M(\ldots)$ & & a $\diva$-transformation of $1$-buyer pricing algorithm to the case of $M$ buyers		\\
		
		\hline
		$\SReg^\ind(\ldots)$ & & individual strategic regret of a dividing algorithm		\\
		\hline
		$\SReg^\dev(\ldots)$& & deviation strategic regret of a dividing algorithm		\\
		
		\hline
		$\sr$ & & a stopping rule used in a	$\diva_M$-transformation of $1$-buyer pricing algorithm		\\
		
		\hline
		$\langle\A\rangle$ & & a transformation of a RPPA algorithm $\A$ s.t.\ all penalization sequences of nodes are replaced by reinforced  penalization ones	\\
		
		\hline
		$\n_i^m$ & & the tracking node of a buyer $m$ by $\diva_M$-transformed RPPA algorithm  in a period $i$ 	for a given play of the game \\

		\hline
	\end{tabular}
\end{table}

%\newpage

\subsection{Notations related to divPRRFES}
See Table~\ref{tbl_notations_divPRRFES}.

\begin{table}
	\centering
	\caption{Notations related to divPRRFES.}
	\label{tbl_notations_divPRRFES}
	%		\vspace{-0.3cm}
	\begin{tabular}{|c|l|p{12cm}|}
		\hline
		\textbf{Notation} & \textbf{Expression} & \textbf{Description} \\
		\hline
		\hline

		$r$ & & the number of penalization rounds (a parameter of PRRFES)		\\
		\hline
		
		$g(l)$ & & the exploitation rate (a parameter of PRRFES)		\\
		\hline
		
		$l$ & & a phase of PRRFES	\\
		\hline
		
		$\ve_l$ & $=2^{-2^l}$& the iteration parameter of a phase $l$	\\
		\hline
		
		$q^m_l$ & & the last accepted price by a buyer $m$ before a phase $l$ 	for a given play of the game	\\
		\hline
		
		$p^m_{l,k}$ & & the $k$-th exploration price of a buyer $m$ in a phase $l$ 	for a given play of the game	\\
		\hline
		
		$K^m_{l}$ & & the last accepted exploration price of a buyer $m$ in a phase $l$ 	for a given play of the game	\\
		\hline
		
		$l^m_i$ & & the current phase of a buyer $m$ in a period $i$ for a given play of the game	\\
		
		\hline
		$l(\n)$ & & the phase of a node $\n$ from the tree of the  algorithm  PRRFES		\\
		\hline
		$q(\n)$ & & the last accepted price before the current phase of a node $\n$ from the tree of the  algorithm  PRRFES		\\
		\hline
	\end{tabular}
\end{table}

\newpage

\section{Discussion \& extensions of the result}

%{\bf Buyer strategies and beliefs.}
%An example of a buyer the may  have worst-case beliefs conditioning by his knowledge: against a div-transformed pricing algorithm, he believes that the valuations of his rivals are higher than his until he becomes a single suspected buyer in $\Sset_i$ (thus, he behaves truthfully in periods $i$ s.t.\ $|\Sset_i| > 1$ and strategically as in RPPA in periods $i$ s.t.\ $|\Sset_i| = 1$). 

{\bf Improvements of divPRRFES.}
For practical use, there are several places where divPRRFES can be improved. 
For instance, (a) the penalization parameter $r$ can be made adaptive to take into account the rounds in which a buyer is eliminated (i.e., reduce the number of penalizations by the number of rivals currently suspected by the seller); 
(b) or the stopping rule $\sr_{\A_1}$ can faster eliminate bidders, since the lower bound $u^m_i$ can be updated each time the buyer $m$ accepts an exploration price $p^m_{l,k}$.
Despite these improvements would require some additional pages in our proofs, they do not improve the asymptotic bound of $O(\log\log T)$.
%, since для доказательства мы во многих местах огрубляли параметры: 

%7. Построили сразу horizon-independent алгоритм
{\bf Horizon independence.} The algorithm divPRRFES is horizon-independent since it is based on the horizon-independent PRRFES $\A_1$, which induces the subalgorithm $\langle\A_1\rangle$ and the stopping rule $\sr_{\A_1}$. Hence, the seller is not required to know in advance the number of rounds~$T$ of the game, when she applies divPRRFES.

\end{document}